\documentclass[11pt,oneside]{article}
\usepackage[margin=1in]{geometry}

\usepackage{amsmath,amsthm}
\usepackage[natbibapa]{apacite}
\DeclareMathOperator*{\argmax}{argmax}

\usepackage{graphicx}
\usepackage{mathtools}
\usepackage{extarrows}
\usepackage{verbatim}
\usepackage[hidelinks]{hyperref}

\usepackage{enumitem}
\usepackage{amssymb} 
\allowdisplaybreaks	

\usepackage{caption}

\numberwithin{equation}{section}
\theoremstyle{plain}                
\newtheorem{theorem}{Theorem}[section]
\newtheorem{lemma}[theorem]{Lemma}

\newtheorem{corollary}[theorem]{Corollary}

\theoremstyle{definition}           
\newtheorem{definition}[theorem]{Definition}

\newtheorem{assumption}[theorem]{Assumption}
\newtheorem{notation}[theorem]{Notation}
\theoremstyle{remark}
\newtheorem{remark}[theorem]{Remark}

\newcommand{\lemref}[1]{Lemma~\ref{#1}}

\usepackage[toc,title,page]{appendix}

\providecommand{\R}{} \renewcommand{\R}{{\mathbb R}}
 
\newcommand{\N}{{\mathbb N}}

\newcommand\ue{\underline{\epsilon}}
\newcommand\be{\overline{\epsilon}}

\newcommand{\cU}{\mathcal{U}}
\newcommand{\eM}{\epsilon_M}
\newcommand{\bx}{\mathbf{x}}
\newcommand{\bp}{\mathbf{p}}
\newcommand{\sn}{\mathrm{sn}}
\newcommand{\D}{\mathbb{D}}

\title{Portfolio Choice under General Utility \\ with Transaction Costs and Search Frictions}

\author{
	\textsc{Tae Ung Gang}
	\thanks{Department of Mathematical Sciences, UNIST, South Korea (E-mail: {\it gahnghansahnae@unist.ac.kr})}
	\and
	\textsc{Donghan Kim}
	\thanks{Department of Mathematical Sciences, KAIST, South Korea (E-mail: {\it kimdonghan@kaist.ac.kr})}
}

\begin{document}

\maketitle

\begin{abstract}
    We study finite-horizon portfolio optimization with proportional transaction costs and trading opportunities arriving at the jump times of a Cox process. Borrowing and short-selling are prohibited, while utility functions need not be concave, increasing, or differentiable. The admissible class includes differentiable utilities with asymptotic elasticity greater than or equal to one. The associated Hamilton--Jacobi--Bellman equation is semi-linear and nonlocal, with control entering only through a zeroth-order term. We show that the normalized equation admits a unique bounded solution that is continuous in all state variables and classical in time and log-price. The proof combines a contraction argument with interior Schauder estimates, and verification identifies the value function and yields an optimal Markovian feedback strategy. Numerical examples with non-concave utilities exhibit, at fixed time and pre-trade wealth, multiple intervals associated with the same trading action, departures from the usual buy--no-trade--sell ordering, and abrupt switches in the optimal post-trade risky weight.
\end{abstract}

\bigskip

\noindent \textbf{Keywords:} Portfolio optimization; non-concave utility; transaction costs; search frictions; HJB equation; optimal trading strategy.

\bigskip

\section{Introduction}

Since the seminal contributions of \citet{Merton1969, Merton1971}, frictionless portfolio optimization has provided a benchmark for continuous-time investment theory. In Merton's frictionless market, an investor with constant relative risk aversion~(CRRA) maintains a constant fraction of wealth in the risky asset, typically requiring continuous rebalancing. Market illiquidity limits the applicability of this prescription. Two prominent sources of illiquidity are proportional transaction costs, which make each trade costly, and search frictions, which restrict trading to times when a counterparty becomes available.

When trading is continuously available but incurs proportional transaction costs, the investor no longer continuously tracks the frictionless target. \citet{MagillConstantinides1976} observed that the optimal policy is described by a no-trade region, within which portfolio adjustment is not optimal. \citet{DavisNorman1990} provided a rigorous analysis of the infinite-horizon consumption--investment problem with power utility using singular stochastic control; their optimal policy keeps the portfolio in the closure of a no-trade wedge through purchases and sales at its boundaries. \citet{ShreveSoner1994} relaxed the technical assumptions of that analysis and used viscosity methods to establish regularity and characterize the optimal policy. Finite-horizon problems were studied by \citet{Liu2002} and \citet{Dai2009}, with the latter developing a parabolic double-obstacle formulation and analyzing the associated free boundaries. \citet{Janecek2004} and \citet{Gerhold2014} developed small-cost asymptotics for trading boundaries and welfare losses. These studies work with concave utilities, often of power or logarithmic type, and establish the buy--no-trade--sell ordering of the optimal trading regions.

A different kind of illiquidity arises when trading is restricted to exogenous, typically random, times. Early contributions consider regularly spaced trading dates \citep{Rogers2001} or the arrival times of a Poisson process \citep{Rogers2002,Matsumoto2006}, the latter obtaining the value function for logarithmic utility together with an asymptotic expansion of the optimal strategy. \citet{Pham2008} studied optimal consumption when the risky asset can be traded and observed only at Poisson times, characterizing the value functions through a coupled system of integro-differential equations of \cite{PT2009}. \citet{Cretarola2011} established $C^1$ regularity for the value functions and derived optimal policies in feedback form. Related developments include finite-horizon problems and convergence to the liquid-market limit \citep{Gassiat2011}, regime-switching liquidity \citep{Gassiat2014}, liquidity breakdowns \citep{Diesinger2010}, and allocation between liquid and illiquid assets \citep{Ang2014}. This literature emphasizes the effect of restricted trading opportunities rather than proportional costs incurred at each trade.

When both frictions are present, the investor can rebalance only at trading opportunities and incurs proportional costs whenever a trade is executed. \citet{Gang2023} characterize the no-trade region for logarithmic utility in a model combining transaction costs and search frictions; for power utility, the same authors develop a unified asymptotic framework \cite{TJ2026} in which both frictions become small, connecting the limiting regimes with transaction costs alone and search frictions alone.

Much of this literature assumes concave utility. Nonconcave objectives nevertheless arise naturally in financial decision-making. Examples include delegated portfolio management with convex incentive compensation \citep{Carpenter2000}, S-shaped value functions of prospect theory \citep{Berkelaar2004, JinZhou2008, KahnemanTversky1979}, aspiration and goal-reaching preferences \citep{AristidouGigaLeeZapatero2025, Browne1999, DiecidueVandeVen2008}, and classical utilities with alternating curvature \citep{FriedmanSavage1948, Markowitz1952}.

In suitable unconstrained complete-market settings, one can study non-concave utility maximization via concavification: under appropriate assumptions, replacing the utility with its concave envelope preserves the optimal value \citep{Reichlin2013}. This replacement need not remain valid under trading constraints or transaction costs. Existing results on non-concave utility maximization under trading constraints or frictions, including those in \cite{DaiKouQianWan2022,QianYang2026}, characterize the value function using viscosity solutions of the HJB equation; to the best of our knowledge, these studies do not establish a classical-solution characterization. The underlying difficulty is structural: when trading is continuously available but incurs proportional transaction costs, the HJB equation takes the form of a variational inequality with gradient constraints, with free boundaries separating the buy, no-trade, and sell regions \citep{Dai2009,ShreveSoner1994}. The classical regularity results cited above exploit the concavity and homothetic structure of power utility functions, while convex duality methods likewise require concavity. Neither approach therefore applies directly to non-concave utilities, and concavification cannot generally be used to overcome this difficulty in the presence of transaction costs, as shown in \cite{QianYang2026}.

The main message of this paper is that search frictions remove this obstruction. Because trades occur only at Cox arrival times, cash and share holdings remain constant between opportunities, while the log-price evolves as a diffusion. The trading decision enters the HJB equation through a nonlocal zeroth-order term: the arrival intensity multiplied by the difference between the best post-trade value and the current value. The HJB equation is therefore semi-linear and uniformly parabolic in the log-price, with the holdings playing the role of parameters. This structure permits classical regularity in the diffusive variables without imposing concavity or homotheticity on utility.

We consider a finite-horizon market with one risk-free asset and one risky asset whose drift and volatility may depend on time and the current price. Purchases and sales incur price-dependent proportional costs, and trading opportunities arrive at the jump times of a Cox process with time- and price-dependent intensity. Borrowing and short-selling are prohibited, and the investor maximizes expected utility of terminal liquidated wealth. For some $p>0$, our only requirement on the utility $U$ is that
\[
    r\longmapsto\frac{U(e^r)}{\mathcal{U}_p(e^r)},\qquad \mathcal{U}_p(w):=1+w^p+w^{-p},
\]
be bounded and H\"older continuous on $\mathbb{R}$. Thus $U$ need not be increasing, concave, or differentiable. The admissible class also includes differentiable utilities with asymptotic elasticity greater than or equal to one; we do not impose the asymptotic elasticity restriction used in the standard concave-utility framework of \citet{Sch1999}.

Our main contributions are as follows.
\begin{itemize}
    \item \emph{Classical solvability.} We prove that the normalized HJB equation admits a unique bounded solution that is continuous in all state variables and classical in time and log-price before maturity. Its derivatives satisfy local parabolic H\"older estimates uniformly over the portfolio positions. The proof combines a Feynman--Kac fixed-point construction, estimates in an intrinsic H\"older seminorm, and interior Schauder estimates for linear parabolic equations.

    \item \emph{Verification and optimal feedback.} We identify both the deterministic value function and the conditional value process and construct an optimal Markovian feedback strategy through a Borel measurable maximizing selector. At each trading opportunity, the optimal decision is obtained by comparing the best attainable post-trade values. The resulting action can be interpreted as a purchase, a sale, or no trade, without requiring a particular ordering or connectedness of the corresponding action sets.

    \item \emph{Qualitative features of optimal trading.} Numerical examples with smoothed aspiration utilities exhibit multiple intervals associated with the same action along fixed-time, fixed-wealth slices, reversals of the usual buy--no-trade--sell ordering, and abrupt switches in the optimal post-trade risky weight. In one example, the computed weight switches directly from zero to one. We prove that any discontinuity of an optimal post-trade risky-weight selector must occur at a state with multiple global maximizers yielding different post-trade weights. For the example with two aspiration levels, the competing decisions can be interpreted as securing one aspiration level while gambling for another.
\end{itemize}

The distinction between regularity in the diffusive and portfolio variables is important. We obtain classical regularity in time and log-price, whereas continuity in the portfolio positions is sufficient for the verification argument because those positions change only at trading opportunities. We also note that the smoothing of the aspiration utilities is essential to our approach; genuine jumps in $U$ violate the H\"older condition, and we do not claim convergence of values or policies as the smoothing vanishes.

The rest of the paper is organized as follows. Section~\ref{sec:model} introduces the market model, and Section~\ref{sec:objective} formulates the optimization problem with admissibility constraints. Section~\ref{sec:HJB} derives the HJB equation, establishes existence, uniqueness, and regularity, and proves the verification theorem. Section~\ref{sec:optimal-portfolio} characterizes optimal trades, while Section~\ref{sec:numerical} presents numerical examples of non-concave utilities and their economic interpretation. Auxiliary lemmas are collected in the Appendix.

\bigskip

\section{Model} \label{sec:model}

Throughout the paper, we fix a finite horizon $T \in (0,\infty)$. Let $(\Omega,\mathcal{F},\mathbb{F}=(\mathcal{F}_t)_{t\in[0, T]},\mathbb{P})$ be a filtered probability space satisfying the usual conditions, and let $B=(B_t)_{t\in[0,T]}$ be a standard $\mathbb{F}$-Brownian motion.

Our financial market consists of one risk-free asset $S^{(0)}$ bearing no interest, whose price is normalized to
\[
    S_t^{(0)} \equiv 1, \qquad t\in[0,T],
\]
and one risky asset $S=(S_t)_{t\in[0,T]}$, whose dynamics are given by
\begin{equation*}
    d S_{t} = S_{t} \Big( \mu(t, S_{t}) d t + \sigma(t, S_{t}) d B_{t} \Big), \qquad t\in[0,T].
\end{equation*}
Here, the real-valued functions $\mu$ and $\sigma$ defined on $[0, T] \times (0, \infty)$ are assumed to be bounded, $(t, x) \mapsto \mu(t, e^{x})$ and $(t, x) \mapsto \sigma(t, e^{x})$ are Lipschitz, and also satisfy
\begin{equation} \label{Assumption_0}
    \inf \big\{ \sigma(t, s) : (t, s) \in [0, T] \times (0, \infty) \big\} > 0.
\end{equation}

We assume $S_0>0$, hence $S_t>0$ for all $t\in[0,T]$ almost surely. By It\^o's lemma, the log-price process $L_t:=\ln S_t$ satisfies
\begin{equation}    \label{eq:logstock}
    dL_t = \Big( \mu(t,e^{L_t})-\frac{\sigma^2(t,e^{L_t})}{2} \Big)dt + \sigma(t,e^{L_t})\,dB_t.
\end{equation}

We consider two sources of market illiquidity:
\begin{enumerate}[label=(\roman*)]
    \item \textbf{[Transaction costs]} Buying the stock at time $t$ incurs proportional cost $\overline\epsilon(S_t)\in(0,1)$, while selling incurs proportional cost $\underline\epsilon(S_t)\in(0,1)$.
    \item \textbf{[Search frictions]} Trading is possible only when an exogenous opportunity arrives. Let $P=(P_t)_{t\in[0,T]}$ be an $\mathbb{F}$-adapted Cox process with continuous compensator
    \[
        A_t:=\int_0^t\lambda(u,S_u)\,du.
    \]
    In particular, $\widetilde{P}_t:=P_t-A_t$ is an $\mathbb{F}$-martingale. Such a process may be constructed as $P_t=\overline{P}_{A_t}$ from a unit-rate Poisson process $\overline{P}$ independent of $B$. Trading opportunities occur only at jump times of $P$; in particular, rebalancing at time $0$ is not allowed.
\end{enumerate}

We shall impose some regularity conditions on the market parameters.
\begin{assumption}  \label{Assumption_1}
    Let $\overline{\epsilon},\underline{\epsilon} \in C^1((0,\infty); (0,1))$ be the transaction cost functions and let $\lambda:[0,T]\times(0,\infty)\to(0,\infty)$ be the intensity function. The maps $l\mapsto\overline{\epsilon}(e^l)$, $l\mapsto\underline{\epsilon}(e^l)$, and $(t,l)\mapsto\lambda(t,e^l)$ are bounded and Lipschitz continuous. We additionally assume that
    \begin{equation*}
        \epsilon_{M} := \sup_{s \in (0, \infty)} \big\{ \max \{ \underline{\epsilon}(s), \overline{\epsilon}(s) \} \big\} < 1 \quad \text{and} \quad
        \Lambda := \sup_{(t, s) \in [0, T] \times (0, \infty)} \lambda(t, s) < \infty;
    \end{equation*}
\end{assumption}

We note that the transaction cost conditions of Assumption \ref{Assumption_1} are automatically satisfied when $\overline{\epsilon}$ and $\underline{\epsilon}$ are constants in $(0,1)$.

For $x\in\mathbb R$, we write $x^+ := \max\{x,0\}$ and
$x^- := \max\{-x,0\}$. A \emph{trading strategy} is a predictable process $N=(N_t)_{t\in[0,T]}$, where $N_t$ denotes the number of stock shares traded at time $t$: $N_t>0$ corresponds to buying, $N_t<0$ to selling, and $N_t=0$ to no trading. 

Let $H_t$ denote the number of stock shares held at time $t$, and let $W_{t}^{(0)}$ denote the wealth invested in the risk-free asset $S^{(0)}$ at time $t$. Since trading is possible only at jump times of $P$, both $H$ and $W^{(0)}$ are pure-jump processes. Their dynamics are
\begin{align}
    \Delta H_t &= N_t\,\Delta P_t, \label{SDE_H}
    \\
    \Delta W_t^{(0)} &=
    \Big( -[1+\overline\epsilon(S_t)]S_tN_t^+ +[1-\underline\epsilon(S_t)]S_tN_t^- \Big)\Delta P_t. \label{SDE_W0}
\end{align}

We denote the wealth invested in the risky asset $S$ by 
\[
    W_t^{(1)} := S_t H_t, \qquad t\in[0,T].
\]
Since $S$ is continuous and $H$ jumps only when $P$ jumps, integration by parts gives
\begin{equation}    \label{eq:W1dyn}
    dW_t^{(1)} = d(S_tH_t) = H_{t-}\,dS_t + S_t\,dH_t
    = W_{t-}^{(1)}\big(\mu(t,S_t)\,dt+\sigma(t,S_t)\,dB_t\big) + S_tN_t\,dP_t.        
\end{equation}

\bigskip

\section{Objective} \label{sec:objective}

We denote by
\[
    \kappa_0 := 0, \qquad \kappa_n := \inf\{t \ge 0:P_t=n\}, \qquad n \in \mathbb{N},
\]
the $n$-th jump time of the Cox process $P$ with convention $\inf\emptyset=\infty$.

We define the total wealth $W$ and the liquidated total wealth $W^{\mathrm{liq}}$ as follows:
\begin{align}   
    W_t &:= W_t^{(0)} + W_t^{(1)} = W_t^{(0)} + S_t H_t,    \label{eq:W}
    \\
    W_t^{\mathrm{liq}} &:= W_t^{(0)} + W_t^{(1)} - \underline{\epsilon}(S_t)W_t^{(1)} = W_t^{(0)} + [1-\underline\epsilon(S_t)]S_t H_t. \nonumber
\end{align}

In our market model, we do not allow borrowing cash and short-selling of the risky asset. Also, we assume that there is no trading at initial time. Therefore, the initial cash and initial number of holding stock shares satisfy the following assumption.

\begin{assumption}[Initial portfolio]   \label{ass:init}
We assume that $S_{0}, W_{0}^{(0)}, H_{0}$ are deterministic and 
\begin{equation} \label{Initial_condition}
        W_0^{(0)} \ge 0, \qquad H_0 \ge 0, \qquad W_0^{\mathrm{liq}} > 0.
    \end{equation}
\end{assumption}

Since both cash and stock holdings are constrained to be nonnegative, their post-trade values must also be nonnegative. To be more specific, let $n \in \mathbb{N}$ be given. Since $\Delta P_{\kappa_{n}}=1$, we have 
\begin{equation} \label{Post_trade_H_W0}
    H_{\kappa_{n}} = H_{\kappa_{n}-}+N_{\kappa_{n}} \quad \text{and} \quad W_{\kappa_{n}}^{(0)} = W_{\kappa_{n}-}^{(0)} - [1+\overline{\epsilon}(S_{\kappa_{n}})]S_{\kappa_{n}} N_{\kappa_{n}}^{+} + [1-\underline{\epsilon}(S_{\kappa_{n}})]S_{\kappa_{n}} N_{\kappa_{n}}^{-}.
\end{equation}
As we want to make $H_{\kappa_{n}} \ge 0$, it is required that $N_{\kappa_{n}} \ge - H_{\kappa_{n}-}$ from the first equality of \eqref{Post_trade_H_W0}. To make $W_{\kappa_{n}}^{(0)} \ge 0$, there are two cases. If $N_{\kappa_{n}}>0$, then $N_{\kappa_{n}}^{+}=N_{\kappa_{n}}$ and $N_{\kappa_{n}}^{-}=0$. So the second equality of \eqref{Post_trade_H_W0} implies $W_{\kappa_{n}}^{(0)} = W_{\kappa_{n}-}^{(0)} - [1+\overline{\epsilon}(S_{\kappa_{n}})]S_{\kappa_{n}} N_{\kappa_{n}}$ and we require 
\[
    N_{\kappa_{n}} \le \frac{W_{\kappa_{n}-}^{(0)}}{[1+\overline{\epsilon}(S_{\kappa_{n}})]S_{\kappa_{n}}}.
\]
If $N_{\kappa_{n}}\le0$, then
    \[
        N_{\kappa_{n}} \le 0 \le \frac{W_{\kappa_{n}-}^{(0)}}{[1+\overline{\epsilon}(S_{\kappa_{n}})]S_{\kappa_{n}}},
    \]
because $W_{\kappa_{n}-}^{(0)}\ge0$. Hence we require
\[
  N_{u} \in \bigg[ - H_{u-}, \frac{W_{u-}^{(0)}}{[1 + \overline{\epsilon}(S_{u})] S_{u}} \bigg]
\]
for all $u = \kappa_{n}$. For technical convenience, we impose the same state-dependent constraint at every time.

\begin{definition}[Admissible strategy] \label{def:admissible}
    A predictable process $N=(N_t)_{t\in[0,T]}$ is called \emph{admissible} if $N_{0} = 0$ and 
    \begin{equation} \label{Admissible}
        \mathbb{P} \bigg( N_{t} \in \bigg[ - H_{t-}, \frac{W_{t-}^{(0)}}{[1 + \overline{\epsilon}(S_{t})]S_{t}} \bigg] \quad \text{for every } t \in (0, T] \bigg) = 1.
    \end{equation}
    We write $\mathcal{N}(0)$ for the class of admissible strategies on $[0,T]$. More generally, when the problem is restarted at time $t$, the corresponding admissible class on $[t,T]$ will be denoted by $\mathcal{N}(t)$, suppressing the dependence on the current state from the notation.
\end{definition}

\begin{lemma}   \label{lem:admissible}
    Let $N\in\mathcal{N}(0)$. Under Assumption \ref{ass:init},
    \begin{enumerate}[label=(\roman*)]
        \item $W_{\kappa_n}^{(0)} \ge 0$ and $H_{\kappa_n} \ge 0$ on $\{\kappa_n \le T\}$ for each $n \in \mathbb{N}$;
        \item $W_{t} \ge W_t^{\mathrm{liq}} > 0$ for all $t\in[0,T]$ a.s.
    \end{enumerate}
\end{lemma}

\begin{proof}
    By Assumption~\ref{ass:init}, the initial cash position and stock holding are nonnegative. Moreover, the interval constraint in Definition~\ref{def:admissible}, together with \eqref{Post_trade_H_W0}, preserves the nonnegativity of both variables at every jump time. Since $W^{(0)}$ and $H$ are constant between consecutive jump times, an induction over the jump times of $P$ proves (i).

    To prove (ii), we first prove that $W_t > 0$ for all $t \in [0, T]$, a.s. By \eqref{Initial_condition},
    \[
        W_0^{\mathrm{liq}} = W_0^{(0)} + [1-\underline{\epsilon}(S_0)]S_0H_0 > 0.
    \]
    Since $W_0^{(0)}\ge0$, $H_0\ge0$, $S_0>0$, and $1-\underline{\epsilon}(S_0)>0$, it follows that $(W_0^{(0)},H_0)\neq(0,0)$. We now show that the portfolio cannot reach $(0,0)$ at a trading time. Let $t\in(0,T]$ be a jump time of $P$, and suppose that $(W_{t-}^{(0)},H_{t-})\neq(0,0)$. Since $\Delta P_t=1$, equations \eqref{SDE_H} and \eqref{SDE_W0} yield
    \begin{align*}
        H_t &= H_{t-}+N_t,
        \\
        W_t^{(0)} &= W_{t-}^{(0)} - [1+\overline{\epsilon}(S_t)]S_tN_t^{+} + [1-\underline{\epsilon}(S_t)]S_tN_t^{-}.
    \end{align*}
    By admissibility, $W_t^{(0)}\ge0$ and $H_t\ge0$. Assume, to obtain a contradiction, that $W_t^{(0)}=H_t=0$. The equality $H_t=0$ implies $N_t=-H_{t-} \le 0$. Therefore, $N_t^{+}=0$, $N_t^{-}=H_{t-}$, and hence
    \[
        0 = W_t^{(0)} = W_{t-}^{(0)} + [1-\underline{\epsilon}(S_t)]S_tH_{t-}.
    \]
    Both terms on the right-hand side are nonnegative, and $[1-\underline{\epsilon}(S_t)]S_t>0$. Consequently, $W_{t-}^{(0)}=0$, $H_{t-}=0$ which contradicts $(W_{t-}^{(0)},H_{t-})\neq(0,0)$. Thus, $(W_t^{(0)},H_t)\neq(0,0)$ after every trading opportunity. Since $(W^{(0)},H)$ is constant between consecutive jump times, induction over the jump times of $P$ gives
    \[
        (W_t^{(0)},H_t)\in[0,\infty)^2\setminus\{(0,0)\}, \qquad \text{for all } t\in[0,T], \text{ a.s.},
    \]
    and we conclude that
    \[
        W_t = W_t^{(0)}+S_tH_t >0, \qquad \text{for all } t\in[0,T], \text{ a.s.}
    \]
    Moreover, since $\underline{\epsilon}(S_t)\le\epsilon_M<1$,
    \begin{equation*}
        W_t^{\mathrm{liq}} = W_t^{(0)} + [1-\underline{\epsilon}(S_t)]S_tH_t \ge (1-\epsilon_M) \big(
        W_t^{(0)}+S_tH_t \big) = (1-\epsilon_M)W_t > 0.
    \end{equation*}
    Clearly, $W_t^{\mathrm{liq}}\le W_t$. Therefore, $W_{t} \ge W_t^{\mathrm{liq}} > 0$ for all $t\in[0,T]$, a.s.
\end{proof}

An investor wishes to maximize expected utility from terminal liquidated wealth:
\begin{equation}    \label{Objective}
    \sup_{N\in\mathcal{N}(0)} \mathbb{E}\Big[ U\big(W_{T}^{\mathrm{liq}}\big) \Big].
\end{equation}
Here $U:(0,\infty)\to \R$ is a utility function, \textit{but none of monotonicity, concavity, or differentiability needs to be assumed}.

For technical reasons, we impose the following assumption on the utility function $U$.

\begin{assumption}[Utility function] \label{U_assumption}
    Let $U:(0,\infty)\to \R$ be the utility function. For $q \in (0,\infty)$, let $\cU_{q}(w) := 1+w^{q}+w^{-q}$ defined on $(0, \infty)$. Then, there exist constants $p > 0$ and $\delta_{U} \in (0, 1]$ such that 
    \begin{align*}
        x \mapsto \frac{U(e^{x})}{\cU_{p}(e^{x})} \text{ is bounded and $\delta_{U}$-H\"older continuous on $\R$.}
    \end{align*}
\end{assumption}

For example, a strictly increasing function that is neither concave nor convex
\begin{equation}    \label{eq:U_counterexample}
    U(w) = w + \frac{1}{2} \sin(w), \qquad w>0,
\end{equation}
satisfies Assumption \ref{U_assumption}. In particular, for differentiable $U$ our framework allows utilities whose asymptotic elasticity
\[
    \limsup_{w\to\infty}\frac{wU'(w)}{U(w)}
\]
is greater than or equal to $1$. The utility in \eqref{eq:U_counterexample} has asymptotic elasticity $3/2$. Thus, we do not impose the asymptotic elasticity restriction used in the concave-utility framework of \citet{Sch1999}.

From now on, we suppose that Assumptions \ref{ass:init} and \ref{U_assumption} hold.

\bigskip

\section{HJB equation}  \label{sec:HJB}

In this section, we derive the dynamic programming equation associated with the optimization problem \eqref{Objective}.

Let $\mathbb{D} := [0,\infty)^{2} \setminus \{ (0,0) \}$. For $(t,l,x,y) \in [0,T] \times \R \times \mathbb{D}$, let $\mathbb{E}_{t,l,x,y}$ denote expectation for the controlled system restarted from 
\[
    (L_{t}, W_{t}^{(0)}, H_{t}) = (l,x,y).
\]
Let $\mathcal{N}(t,l,x,y)$ denote the class of admissible controls on $[t,T]$ for this restarted system. When the dependence on the initial state is clear, we simply write $\mathcal{N}(t)$. We use $W^{\mathrm{liq},N}$ when it is necessary to indicate the dependence of the wealth process on the control $N$.

We define the deterministic Markov value function by
\begin{equation} \label{Value_define}
    V(t,l,x,y) := \sup_{N \in \mathcal{N}(t,l,x,y)} \mathbb{E}_{t,l,x,y} \Big[ U \big( W_{T}^{\mathrm{liq}, N} \big) \Big], \qquad (t,l,x,y) \in [0,T] \times \R \times \mathbb{D}.
\end{equation}
For simplicity, we write $\cU := \cU_{p}$. As in Assumption \ref{U_assumption}, we introduce the normalized value function
\begin{equation*}
    v(t,l,x,y) := \frac{V(t,l,x,y)}{\cU \big( x + y e^{l} \big)}.
\end{equation*}

At a trading opportunity, the investor chooses a trade size obtained from \eqref{Admissible}
\[
    n \in \mathcal{M}(l,x,y) := \bigg[ -y, \frac{x e^{-l}}{1 + \overline{\epsilon}(e^{l})} \bigg].
\]

For a test function $\phi$, define the post-trade intervention operator by
\begin{equation*}
    (\mathcal{J}\phi)(t,l,x,y;n) :=
    \begin{cases}
        \phi\big( t,l, x - [1+\overline\epsilon(e^{l})] e^{l} n, y + n \big), & n \ge 0,
        \\[0.4em]
        \phi\big( t,l, x - [1-\underline\epsilon(e^{l})] e^{l} n, y + n \big), & n < 0.
    \end{cases}
\end{equation*}

Between trading times, the variables $x$ and $y$ remain frozen, while the log-price $L = \ln S$ evolves as in \eqref{eq:logstock}, hence the local diffusion operator acting on the $l$-variable is defined as
\begin{equation*}
    (\mathcal{L}\phi)(t,l,x,y) := \bigg( \mu(t,e^{l}) - \frac{\sigma^{2}(t,e^{l})}{2} \bigg) \frac{\partial \phi}{\partial l}(t,l,x,y) + \frac{\sigma^{2}(t,e^{l})}{2} \frac{\partial^{2}\phi}{\partial l^{2}}(t,l,x,y).
\end{equation*}

Therefore, the dynamic programming principle formally suggests that $V$ satisfies
\begin{equation*}
    \begin{cases}
    0 = \dfrac{\partial V}{\partial t}(t,l,x,y) + (\mathcal{L}V)(t,l,x,y) - \lambda(t,e^{l}) V(t,l,x,y) 
    \\
    \hspace{2.5cm} + \lambda(t,e^{l}) \displaystyle\sup_{n \in \mathcal{M}(l,x,y)} (\mathcal{J}V)(t,l,x,y;n),
    \hfill \qquad \qquad (t,l,x,y) \in [0,T) \times \R \times \mathbb{D},
    \\[0.8em]
    V(T,l,x,y) = U \big( x + [1-\underline{\epsilon}(e^{l})] y e^{l} \big), \hfill \qquad \qquad  (l,x,y) \in \R \times \mathbb{D}.
    \end{cases}
\end{equation*}
In what follows, we normalize $V$ in order to work with a bounded object and to apply the fixed-point/PDE estimates more conveniently. Hence, we formally derive the normalized equation for $\nu$ of 
\begin{equation*}
    \nu(t,l,x,y) := e^{\alpha t} v(t,l,x,y)
\end{equation*}
for a fixed constant $\alpha \ge 0$. For a test function $\phi$, define the normalized intervention operator by
\begin{equation*}
    (\mathcal{K}\phi)(t,l,x,y;n) :=
    \begin{cases}
        \phi\big( t,l, x - [1+\overline\epsilon(e^{l})] e^{l} n, y + n \big)
        \dfrac{\cU \big( x + y e^{l}-\overline{\epsilon}(e^{l}) e^{l}n \big)}{\cU \big( x + y e^{l} \big)}, & n \ge 0,
        \\[1.0em]
        \phi\big( t,l, x - [1-\underline\epsilon(e^{l})] e^{l} n, y + n \big) \dfrac{\cU\big( x + y e^{l}+\underline{\epsilon}(e^{l})e^{l}n \big)}{\cU\big( x + y e^{l} \big)}, & n < 0.
    \end{cases}
\end{equation*}

Set
\begin{align}
    \widetilde{\mu}(t,l,x,y) &:= \mu(t,e^{l}) - \frac{\sigma^{2}(t,e^{l})}{2} + \frac{\sigma^{2}(t,e^{l}) y e^{l} \cU'\big( x + y e^{l} \big)}{\cU \big( x + y e^{l} \big)}, \label{eq:mu_tilde}
    \\
    \gamma(t,l,x,y) &:= \frac{\sigma^{2}(t,e^{l}) \big[ y e^{l} \cU' \big( x + y e^{l} \big) + (y e^{l})^{2} \cU'' \big( x + y e^{l} \big) \big]}{2\cU \big( x + y e^{l} \big)} + \bigg( \mu(t,e^{l}) - \frac{\sigma^{2}(t,e^{l})}{2} \bigg) \frac{ y e^{l} \cU' \big( x + y e^{l} \big)}{\cU \big( x + y e^{l} \big)}.       \label{eq:gamma_define}
\end{align}
To handle regularity of $v$, reparameterize the admissible trade size by setting 
\begin{equation*}
    n = \theta(l,x,y,z) := z \bigg( \frac{x e^{-l}}{1 + \overline{\epsilon}(e^{l})} + y \bigg) - y.
\end{equation*}
We observe that 
\begin{equation*}
    n \in \mathcal{M}(l,x,y) \quad \Longleftrightarrow \quad z \in [0,1].
\end{equation*}
We further define the notation
\[
  K(\phi)(t,l,x,y,z) := (\mathcal{K}\phi)(t,l,x,y;\theta(l,x,y,z))
\]
for a test function $\phi$. Then the transformed function $\nu$ formally solves
\begin{equation} \label{eq:HJB_nu}
    \begin{cases}
    0 = \dfrac{\partial \nu}{\partial t}(t,l,x,y) + \widetilde{\mu}(t,l,x,y) \dfrac{\partial \nu}{\partial l}(t,l,x,y) + \dfrac{\sigma^{2}(t,e^{l})}{2} \dfrac{\partial^{2}\nu}{\partial l^{2}}(t,l,x,y)
    \\[0.8em]
    \qquad + \big( -\alpha - \lambda(t, e^{l}) + \gamma(t,l,x,y) \big) \nu(t,l,x,y)
    \\[0.8em]
    \qquad + \lambda(t, e^{l}) \displaystyle\sup_{z \in [0,1]} K(\nu)(t,l,x,y,z), \hfill (t,l,x,y) \in [0,T) \times \R \times \mathbb{D},
    \\[0.8em]
    \nu(T,l,x,y) = e^{\alpha T} g(l,x,y), \hfill (l,x,y) \in \R \times \mathbb{D},
    \end{cases}
\end{equation}
where 
\[
  g(l,x,y) := \frac{U \big( x+[1-\underline{\epsilon}(e^{l})] y e^{l} \big)}{\cU \big( x + y e^{l} \big)}.
\]

The following lemma guarantees the regularity of a solution for \eqref{eq:HJB_nu}. In what follows, we use the following notations for each set $A$ 
\[
    \| f \|_{\infty; A} := \sup_{(u, \mathbf{y}) \in A} |f(u, \mathbf{y})|, \qquad [f]_{\frac{\delta}{2}, \delta; A} := \sup_{\substack{(u_{1}, \mathbf{y}_{1}), (u_{2}, \mathbf{y}_{2}) \in A \\ (u_{1}, \mathbf{y}_{1}) \neq (u_{2}, \mathbf{y}_{2})}} \frac{|f(u_{1}, \mathbf{y}_{1}) - f(u_{2}, \mathbf{y}_{2})|}{|u_{1} - u_{2}|^{\frac{\delta}{2}} + |\mathbf{y}_{1} - \mathbf{y}_{2}|^{\delta}},
\]
\[
    \| f \|_{C^{1+\frac{\delta}{2},2+\delta}(A)} := \bigg[ \frac{\partial f}{\partial t} \bigg]_{\frac{\delta}{2}, \delta; A} + \bigg[ \frac{\partial^{2} f}{\partial z^{2}} \bigg]_{\frac{\delta}{2}, \delta; A} + \| f \|_{\infty; A} + \bigg\| \frac{\partial f}{\partial t} \bigg\|_{\infty; A} + \bigg\| \frac{\partial f}{\partial z} \bigg\|_{\infty; A} + \bigg\| \frac{\partial^{2} f}{\partial z^{2}} \bigg\|_{\infty; A},
\]
\begin{align*}
    C^{1+\frac{\delta}{2},2+\delta}(A) &:= \big\{ A \ni (t,z) \mapsto f : \| f \|_{C^{1+\frac{\delta}{2},2+\delta}(A)} < \infty \big\},
    \\
    C_{\mathrm{loc}}^{n+\frac{\delta}{2},2n+\delta}(A) &:= \big\{ f : f|_{D} \in C^{n+\frac{\delta}{2},2n+\delta}(D) \text{ for any compact } D \subset A \big\}, \quad n \in \{ 0,1 \},
    \\
    C_{bc}(A) &:= \{ f : A \to \R : f \text{ is bounded and continuous} \}.
\end{align*}
For simplicity, we will suppress the dependency of set $A$ whenever no confusion arises.

\begin{lemma} \label{lem:value_regularity}
    Let $\delta \in (0,\delta_{U}] \cap (0, 1)$ be given. Suppose that $\alpha$ is chosen sufficiently large. Then there exists a unique solution $\nu$ to \eqref{eq:HJB_nu} such that
    \begin{equation}    \label{nu_regularity_1}
        \nu(\cdot,\cdot,x,y) \in C_{\mathrm{loc}}^{1+\frac{\delta}{2},2+\delta}([0,T) \times \R) \quad \text{ for each } (x,y) \in \mathbb{D}, \qquad \nu \in C_{bc}([0,T] \times \R \times \mathbb{D}),
    \end{equation}
    and 
    \begin{equation} \label{nu_regularity_2}
        \sup_{(x, y) \in \D} \| \nu(\cdot, \cdot, x, y) \|_{C^{1+\frac{\delta}{2},2+\delta} ([0,T-\xi] \times [-R,R])} < \infty \quad \text{ for each } (\xi, R) \in (0, T) \times (0, \infty).
    \end{equation}
\end{lemma}

\begin{proof}
In this proof, we shall use the notation of the vector form $\mathbf{x} := (l,x,y)$. Given $t \in [0, T]$, define the process $L_{u}^{(t, \mathbf{x})}$ for $u \ge t$ by 
\begin{align} \label{L_process}
  & d L_{u}^{(t, \mathbf{x})} = \widetilde{\mu}(u,L_{u}^{(t, \mathbf{x})},x,y) d u + \sigma(u, e^{L_{u}^{(t, \mathbf{x})}}) d B_{u} \quad \text{ with } \; L_{t}^{(t, \mathbf{x})} = l
\end{align}
where the strong solution for \eqref{L_process} is uniquely determined due to \lemref{lem:coeff}. Considering \eqref{eq:HJB_nu} and Feynman-Kac formula, we want to prove $\psi(\nu) = \nu$, where the map $\psi$ on $C_{bc}([0,T] \times \R \times \mathbb{D})$ equipped with $\| f \|_{\infty} := \sup_{(s, \mathbf{y}) \in [0,T] \times \R \times \mathbb{D}} |f(s, \mathbf{y})|$ is defined by 
\begin{align}
  & \psi(w)(t, \mathbf{x}) = \mathbb{E} \bigg[ \exp \bigg( \int_{t}^{T} \Big[ - \alpha - \lambda \Big( a, e^{L_{a}^{(t, \mathbf{x})}} \Big) + \gamma(a, L_{a}^{(t, \mathbf{x})},x,y) \Big] d a \bigg) e^{\alpha T} g \Big( L_{T}^{(t, \mathbf{x})}, x, y \Big) \nonumber \\
  & \qquad + \int_{t}^{T} \exp \bigg( \int_{t}^{u} \Big[ - \alpha - \lambda \Big( a, e^{L_{a}^{(t, \mathbf{x})}} \Big) + \gamma(a, L_{a}^{(t, \mathbf{x})},x,y) \Big] d a \bigg) \nonumber \\
  & \qquad \qquad \cdot \lambda \Big( u, e^{L_{u}^{(t, \mathbf{x})}} \Big) \sup_{z \in [0,1]} K(w) \big( u, L_{u}^{(t, \mathbf{x})}, x, y, z \big) d u \bigg]. \label{contraction_map}
\end{align}
First recall from Notation \ref{not:master} the barycentric coordinate $b(l,x,y) = \frac{y e^{l} [1 + \overline{\epsilon}(e^{l})]}{y e^{l} [1 + \overline{\epsilon}(e^{l})] + x}$ and observe that 
\begin{align*}
    & z < b(L,x,y) \; \Longleftrightarrow \; \theta(L,x,y,z) \in [-y, 0); \\
    & z \ge b(L,x,y) \; \Longleftrightarrow \; \theta(L,x,y,z) \in \bigg[ 0,\frac{x e^{-L}}{1 + \overline{\epsilon}(e^{L})} \bigg],
\end{align*}
so that, by (iii) of \lemref{lem:post_trade}, 
\begin{equation} \label{U_term_bounded}
  \begin{split}
    & \bigg| \frac{\cU \big( x + y e^{L} - \overline{\epsilon}(e^{L}) e^{L} \theta(L,x,y,z) \big)}{\cU \big( x + y e^{L} \big)} \mathbf{1}_{ \{ z \ge b(L,x,y) \} } \bigg| \le c_{\cU}; \\
    & \bigg| \frac{\cU \big( x + y e^{L} + \underline{\epsilon}(e^{L}) e^{L} \theta(L,x,y,z) \big)}{\cU \big( x + y e^{L} \big)} \mathbf{1}_{ \{ z < b(L,x,y) \} } \bigg| \le c_{\cU},
  \end{split}
\end{equation}
where $c_{\cU}$ is the constant from \lemref{lem:post_trade}. We restrict $\alpha > c_{0} := 2 \Lambda c_{\cU} + \Lambda + \| \gamma \|_{\infty}$.

The proof of the lemma proceeds in three steps. \bigskip

\textbf{Step 1. Find fixed point of $\psi$ on $[0, T] \times \R \times \mathbb{D}$}

(i) First, we show that $\psi$ maps $C_{bc}([0,T] \times \R \times \mathbb{D})$ into itself.
\begin{align}
  \big|\psi(w)(t,\mathbf{x})\big| & \le e^{\alpha T} \mathbb{E} \Big[ \Big| g \Big( L_{T}^{(t, \mathbf{x})}, x, y \Big) \Big| \Big] \nonumber \\
  & \quad + \| w \|_{\infty} \, \mathbb{E} \bigg[ \int_{t}^{T} \exp \bigg( \int_{t}^{u} \Big[ - \alpha - \lambda \Big( a, e^{L_{a}^{(t, \mathbf{x})}} \Big) + \gamma(a, L_{a}^{(t, \mathbf{x})}, x, y) \Big] d a \bigg) \nonumber \\
  & \qquad \cdot \lambda \Big( u, e^{L_{u}^{(t, \mathbf{x})}} \Big) \sup_{z \in [0,1]} \bigg\{ \frac{\cU \big( x + y e^{L} - \overline{\epsilon}(e^{L}) e^{L} \theta(L,x,y,z) \big)}{ \cU \big( x + y e^{L} \big)} \mathbf{1}_{ \{ z \ge b(L,x,y) \} } \nonumber \\
  & \qquad \quad + \frac{\cU \big( x + y e^{L} + \underline{\epsilon}(e^{L}) e^{L} \theta(L,x,y,z) \big)}{\cU \big( x + y e^{L} \big)} \mathbf{1}_{ \{ z < b(L,x,y) \} } \bigg|_{L = L_{u}^{(t, \mathbf{x})}} d u \bigg] \\
  & \le e^{\alpha T} C_{*} + \sup_{u \in [0, T]} \mathbb{E} \bigg[ \int_{u}^{T} e^{(- \alpha + \Lambda + \| \gamma \|_{\infty}) (z - u)} \Lambda c_{\cU} d z \bigg] \| w \|_{\infty} \label{phi_upper_bound} \\
  & \le e^{\alpha T} C_{*} + \frac{1}{2} \| w \|_{\infty}, \nonumber
\end{align}
where \eqref{phi_upper_bound} holds due to \eqref{U_term_bounded} and (iv) of \lemref{lem:coeff}. Taking the supremum over $(t,\mathbf{x})$ yields
\begin{equation}    \label{phi_bounded}
    \|\psi(w)\|_\infty \le e^{\alpha T} C_{*} + \frac{1}{2} \| w \|_{\infty}.
\end{equation}
Moreover, we can check first that the map $(l,x,y) \mapsto g(l,x,y)$ is bounded and continuous on $\R \times \mathbb{D}$. Second, the dynamics \eqref{L_process} with assumption $t_{1} < t_{2} \le u$ (the case of $t_{2} < t_{1} \le u$ is shown similarly) implies that there exists a generic constant $C>0$ such that 
\begin{align*}
    \mathbb{E} \bigg[ \Big( L_{u}^{(t_{1}, \mathbf{x}_{1})} - L_{u}^{(t_{2}, \mathbf{x}_{2})} \Big)^{2} \bigg] & \le C (l_{1} - l_{2})^{2} + C \mathbb{E} \bigg[ \bigg( \int_{t_{2}}^{u} \Big( \widetilde{\mu} \big( s,L_{s}^{(t_{1}, \mathbf{x}_{1})},x_{1},y_{1} \big) - \widetilde{\mu} \big( s,L_{s}^{(t_{2}, \mathbf{x}_{2})},x_{2},y_{2} \big) \Big) d s \bigg)^{2} \bigg] \nonumber \\
    & \quad + C \mathbb{E} \bigg[ \bigg( \int_{t_{2}}^{u} \Big( \sigma \Big( s, e^{L_{s}^{(t_{1}, \mathbf{x}_{1})}} \Big) - \sigma \Big( s, e^{L_{s}^{(t_{2}, \mathbf{x}_{2})}} \Big) \Big) d B_{s} \bigg)^{2} \bigg] \nonumber \\
    & \quad + C \mathbb{E} \bigg[ \bigg( \int_{t_{1}}^{t_{2}} \widetilde{\mu} \big( s,L_{s}^{(t_{1}, \mathbf{x}_{1})},x_{1},y_{1} \big) d s \bigg)^{2} \bigg] + C \mathbb{E} \bigg[ \bigg( \int_{t_{1}}^{t_{2}} \sigma \Big( s,e^{L_{s}^{(t_{1}, \mathbf{x}_{1})}} \Big) d B_{s} \bigg)^{2} \bigg] \\
    & \le C \big[ \rho \big( (t_{1}, \mathbf{x}_{1}), (t_{2}, \mathbf{x}_{2}) \big) \big]^{2} + C \int_{t_{2}}^{u} \mathbb{E} \bigg[ \Big( L_{s}^{(t_{1}, \mathbf{x}_{1})} - L_{s}^{(t_{2}, \mathbf{x}_{2})} \Big)^{2} \bigg] d s
\end{align*}
where Cauchy-Schwarz inequality is used in the first inequality, and It\^o's isometry, the fact that $(t, x) \mapsto \mu(t, e^{x})$ and $(t, x) \mapsto \sigma(t, e^{x})$ are bounded Lipschitz, (iii) of \lemref{lem:coeff}, and (i) of \lemref{lem:random} are used in the last inequality.

Now the Gronwall inequality gives 
\begin{equation} \label{L_expectation_Holder}
  \mathbb{E} \big| L_{u}^{(t_{1}, \mathbf{x}_{1})} - L_{u}^{(t_{2}, \mathbf{x}_{2})} \big| \le C \rho \big( (t_{1}, \mathbf{x}_{1}), (t_{2}, \mathbf{x}_{2}) \big)
\end{equation}
for some generic constant $C > 0$. Furthermore, the continuity of $(t,l,x,y,z) \mapsto K(w)(t,l,x,y,z)$ implies the continuity of $(t,l,x,y) \mapsto \sup_{z \in [0,1]} K(w)(t,l,x,y,z)$ by the maximum theorem. For any sequence $(t_n,\mathbf{x}_n)\to(t,\mathbf{x})$, the stability estimate \eqref{L_expectation_Holder}, the continuity of $g$ and $\sup_{z\in[0,1]}K(w)$, and the bounds in \eqref{U_term_bounded} allow us to apply the standard subsequence argument and dominated convergence in \eqref{contraction_map}. Hence $\psi(w)(t_n,\mathbf{x}_n)\to\psi(w)(t,\mathbf{x})$, proving that $\psi(w)$ is continuous.

\medskip

(ii) Second, we show that $\psi$ is a contraction.

By similar way as \eqref{phi_bounded}, for any $w_{1}, w_{2} \in C_{bc}([0, T] \times \R \times \mathbb{D})$ 
\begin{align}
  & \| \psi(w_{1}) - \psi(w_{2}) \|_{\infty} \le \frac{1}{2} \| w_{1} - w_{2} \|_{\infty}. \label{phi_contraction}
\end{align}

Hence, the Banach fixed point theorem says that there exists a unique $\hat{w} \in C_{bc}([0, T] \times \R \times \mathbb{D})$ such that $\psi(\hat{w}) = \hat{w}$.

\bigskip

\textbf{Step 2. Extend regularity of the fixed point to locally H\"older spaces}

Let $\tilde K(f)(t,l,x,y) := \sup_{z \in [0,1]} K(f)(t,l,x,y,z)$. Next, we want to show that 
\begin{equation*}
  \big\| \tilde K(\hat w) \big\|_{\infty} + \big[ \tilde K(\hat w) \big]_{\sn} \le 2 c_{\cU} c_{\alpha} + C (2 C(\alpha) + 4 c_{\alpha}) < \infty
\end{equation*}
where given $t_{i} \in [0,T]$ and $\mathbf{x}_{i} = (l_{i}, x_{i}, y_{i}) \in \R \times \mathbb{D}$ for each $i \in \{ 1, 2 \}$ a new seminorm $[\cdot]_{\sn}$ is defined as 
\begin{align*}
    [f]_{\sn} & := \sup \Bigg\{ \bigg| \frac{f(t_{1}, \mathbf{x}_{1}) - f(t_{2}, \mathbf{x}_{2})}{\big[ \rho \big( (t_{1}, \mathbf{x}_{1}), (t_{2}, \mathbf{x}_{2}) \big) \big]^{\delta}} \bigg| : 0 < \rho \big( (t_{1}, \mathbf{x}_{1}), (t_{2}, \mathbf{x}_{2}) \big) \le 1 \Bigg\}, \qquad \text{ and } \\
    \rho \big( (t_{1}, \mathbf{x}_{1}), (t_{2}, \mathbf{x}_{2}) \big) & := \sqrt{|t_{1} - t_{2}|} + |l_{1} - l_{2}| + \frac{|x_{1} - x_{2}| + \max \big\{ e^{l_{1}}, e^{l_{2}} \big\} |y_{1} - y_{2}|}{\min \big\{ x_{1} + y_{1} e^{l_{1}}, x_{2} + y_{2} e^{l_{2}} \big\}}.
\end{align*}
We construct a sequence of functions $\{ w_{n}(t, \mathbf{x}) \}_{n \in \mathbb{N}} \subseteq C_{bc}([0, T] \times \R \times \mathbb{D})$ starting at $w_{0}(t, \mathbf{x}) = 0$ satisfying 
\begin{align*}
  & w_{n + 1}(t, \mathbf{x}) = \mathbb{E} \bigg[ \exp \bigg( \int_{t}^{T} \Big[ - \alpha - \lambda \Big( a, e^{L_{a}^{(t, \mathbf{x})}} \Big) + \gamma(a, L_{a}^{(t, \mathbf{x})},x,y) \Big] d a \bigg) e^{\alpha T} g \Big( L_{T}^{(t, \mathbf{x})}, x, y \Big) \nonumber \\
  & \qquad + \int_{t}^{T} \exp \bigg( \int_{t}^{u} \Big[ - \alpha - \lambda \Big( a, e^{L_{a}^{(t, \mathbf{x})}} \Big) + \gamma(a, L_{a}^{(t, \mathbf{x})},x,y) \Big] d a \bigg) \nonumber \\
  & \qquad \qquad \cdot \lambda \Big( u, e^{L_{u}^{(t, \mathbf{x})}} \Big) \sup_{z \in [0,1]} K(w_{n}) \big( u, L_{u}^{(t, \mathbf{x})}, x, y, z \big) d u \bigg].
\end{align*}
Suppose $\alpha \ge c_{0}$ so that $\lim_{n \to \infty} w_{n}(t, \mathbf{x}) = \hat{w}(t, \mathbf{x})$ uniformly on $(t, \mathbf{x}) \in [0, T] \times \R \times \mathbb{D}$ (implying $\lim_{n \to \infty} \| w_{n} \|_{\infty} = \| \hat{w} \|_{\infty}$).

Therefore, 
\begin{equation*}
    w_{1}(t, \mathbf{x}) = \mathbb{E} \bigg[ \exp \bigg( \int_{t}^{T} \Big[ - \alpha - \lambda \Big( a, e^{L_{a}^{(t, \mathbf{x})}} \Big) + \gamma(a, L_{a}^{(t, \mathbf{x})},x,y) \Big] d a \bigg) e^{\alpha T} g \Big( L_{T}^{(t, \mathbf{x})}, x, y \Big) \bigg]
\end{equation*}
with $[w_{1}]_{\sn} < \infty$ due to \lemref{lem:pathfun} and (iv) of \lemref{lem:coeff}.

We easily find a positive constant $c_{\alpha}$ which may depend on $\alpha$ such that 
\begin{equation*}
  \| w_{1} \|_{\infty} =: c_{\alpha} < \infty
\end{equation*}
and thus \eqref{phi_contraction} implies 
\begin{equation} \label{w_n_boundedness}
  \begin{split}
  & \| w_{n} \|_{\infty} \le \| w_{1} \|_{\infty} + \frac{1}{2} \| w_{n - 1} \|_{\infty} \; \Longrightarrow \; \| w_{n} \|_{\infty} \le \| w_{1} \|_{\infty} \sum_{k = 0}^{n - 1} 2^{- k} \; \Longrightarrow \; \sup_{n \in \mathbb{N}} \| w_{n} \|_{\infty} \le 2 c_{\alpha} \\
  & \text{and thus } \| \hat{w} \|_{\infty} \le 2 c_{\alpha}.
  \end{split}
\end{equation}
For simplicity, we use Notation~\ref{not:master}. Fix two points $(t_i,\mathbf{x}_i)\in[0,T)\times\mathbb{R}\times\mathbb{D}$, $i\in\{1,2\}$, with
\[
    \rho\big((t_1,\mathbf{x}_1),(t_2,\mathbf{x}_2)\big)\le1,
\]
and assume without loss of generality that $t_1\le t_2$; the opposite case is symmetric. One can observe that 
\begin{align*}
  | w_{n + 1}(t_{1}, &\mathbf{x}_{1}) - w_{n + 1}(t_{2}, \mathbf{x}_{2}) | \\
  & \le | w_{1}(t_{1}, \mathbf{x}_{1}) - w_{1}(t_{2}, \mathbf{x}_{2}) | + \bigg| \mathbb{E} \bigg[ \int_{t_1}^{T}\!A_1(u)\,K^{(n)}_1(u)\,du \bigg] -\mathbb{E} \bigg[ \int_{t_2}^{T}\!A_2(u)\,K^{(n)}_2(u)\,du \bigg] \bigg| \\
  & \le | w_{1}(t_{1}, \mathbf{x}_{1}) - w_{1}(t_{2}, \mathbf{x}_{2})| + \underbrace{\mathbb{E} \bigg[ \int_{t_2}^{T}\!A_1(u)\, \big|K^{(n)}_1(u)-K^{(n)}_2(u)\big|\,du \bigg]}_{=:(\mathrm{I})} \nonumber \\ 
  & \qquad + \underbrace{\mathbb{E} \bigg[ \int_{t_2}^{T}\!\big|A_1(u)-A_2(u)\big|\,
      \big|K^{(n)}_2(u)\big|\,du \bigg]}_{=:(\mathrm{II})}
 +\underbrace{\mathbb{E} \bigg[ \int_{t_1}^{t_2}\!A_1(u)\,
      \big|K^{(n)}_1(u)\big|\,du \bigg]}_{=:(\mathrm{III})}.
\end{align*}
We can evaluate 
\begin{align}
  (\mathrm{I}) & \le \Lambda \int_{t_{2}}^{T} e^{(-\alpha+\Lambda+\|\gamma\|_\infty)(u-t_1)} \, \mathbb{E} \Big[ \Big| K^{(n)}_{1}(u) - K^{(n)}_{2}(u) \Big| \Big] \, du \le \frac{\Lambda \, C}{\alpha - \Lambda - \|\gamma\|_{\infty}} \big([w_{n}]_{sn} + 2 \|w_n\|_{\infty} \big) \, \rho^{\delta}, \label{eq:termI} \\
  (\mathrm{II}) & \le c_{\cU} \, \|w_n\|_{\infty} \, \mathbb{E} \bigg[ \int_{t_{2}}^{T} \big| A_{1}(u) - A_{2}(u) \big| \,du \bigg] \nonumber \\
  & \le c_{\cU} \, \|w_n\|_{\infty} \int_{t_{2}}^{T} \Big( e^{(\Lambda + \|\gamma\|_{\infty}) T} \, C \, \mathbb{E} \big[ \big| L^{1}_{u} -L^{2}_{u} \big| \big] + \Lambda \, \mathbb{E} \big[ \big| \mathcal{E}_{1}(u)-\mathcal{E}_{2}(u) \big| \big] \Big) du \le C_{\alpha} T \rho^{\delta} \label{eq:termII}, \\
  (\mathrm{III}) & \le \Lambda e^{\|\gamma\|_{\infty} T} c_{\cU} \|w_{n}\|_{\infty} |t_{1} - t_{2}| \le \Lambda e^{\|\gamma\|_{\infty} T} c_{\cU} \|w_n\|_{\infty} \rho^{\delta}, \label{eq:termIII}
\end{align}
where \eqref{eq:termI} is derived from \eqref{eq:patchedK} together with (ii) of \lemref{lem:random}, \eqref{eq:termII} is from \lemref{lem:pathfun} and \eqref{L_expectation_Holder}, and \eqref{eq:termIII} is from the mean value theorem. All in all, we deduce 
\begin{equation} \label{w_n_in_Holder}
  [w_{n+1}]_{\sn} \le C(\alpha) + \frac{\tilde{C}}{\alpha} [w_{n}]_{\sn}
\end{equation}
where constant $\tilde C > 0$ is independent of $\alpha$.

We take $\alpha \ge \max \{ c_{0}, 2 \tilde{C} \}$. Then \eqref{w_n_in_Holder} implies that for any $n \in \mathbb{N}$, 
\begin{equation*}
  [w_{n}]_{\sn} \le C(\alpha) \sum_{k=0}^{n-1} 2^{-k} = 2 (1 - 2^{-n}) C(\alpha).
\end{equation*}
Hence, $[\hat{w}]_{\sn} \le 2 C(\alpha)$. Due to \eqref{w_n_boundedness} and \lemref{lem:patchedK} applied to $f = \hat{w}$, we have 
\begin{equation} \label{tilde_K_sn_finite}
    \begin{split}
    & \big\| \tilde{K}(\hat{w}) \big\|_{\infty} \le c_{\cU} \| \hat{w} \|_{\infty} \le 2 c_{\alpha} c_{\cU}, \\
    & \big[ \tilde{K}(\hat{w}) \big]_{\sn} \le C ( [\hat{w}]_{\sn} + 2 \Vert \hat{w} \Vert_{\infty} ) \le C (2 C(\alpha) + 4 c_{\alpha}) < \infty.
    \end{split}
\end{equation}

\bigskip

\textbf{Step 3. Use general PDE theory}

Fix $(x,y) \in \D$ and consider the linear PDE 
\begin{equation} \label{Linear_PDE}
    \begin{cases}
    0 = \dfrac{\partial \nu}{\partial t}(t,l,x,y) + \widetilde{\mu}(t,l,x,y) \dfrac{\partial \nu}{\partial l}(t,l,x,y) + \dfrac{\sigma^{2}(t,e^{l})}{2} \dfrac{\partial^{2}\nu}{\partial l^{2}}(t,l,x,y) + f_{x,y}(t,l)
    \\[0.8em]
    \qquad + \big( -\alpha - \lambda(t, e^{l}) + \gamma(t,l,x,y) \big) \nu(t,l,x,y), \hfill (t,l,x,y) \in [0,T) \times \R \times \mathbb{D},
    \\[0.8em]
    \nu(T,l,x,y) = e^{\alpha T} g(l,x,y), \hfill (l,x,y) \in \R \times \mathbb{D},
    \end{cases}
\end{equation}
where 
\[
    f_{x,y}(t,l) := \lambda(t, e^{l}) \sup_{z \in [0,1]} K(\hat{w})(t,l,x,y,z).
\]
In this step, we use general PDE theory in \cite{AvnerFriedman1964}, where H\"older norm is defined as 
\[
  \sup_{(u_{1}, \mathbf{y}_{1}) \neq (u_{2}, \mathbf{y}_{2})} \frac{|f(u_{1}, \mathbf{y}_{1}) - f(u_{2}, \mathbf{y}_{2})|}{(|u_{1} - u_{2}| + |\mathbf{y}_{1} - \mathbf{y}_{2}|^{2})^{\frac{\delta}{2}}}
\]
in contrast to $[f]_{\frac{\delta}{2}, \delta}$. Nevertheless, this difference does not make any issues because the Cauchy-Schwarz inequality and subadditivity of concave function imply 
\[
  (|u_{1} - u_{2}| + |\mathbf{y}_{1} - \mathbf{y}_{2}|^{2})^{\frac{\delta}{2}} \le 2^{\frac{\delta}{2}} (|u_{1} - u_{2}|^{\frac{\delta}{2}} + |\mathbf{y}_{1} - \mathbf{y}_{2}|^{\delta}),
\]
and increasing function of $x \mapsto x^{\frac{\delta}{2}}$ implies 
\[
   |u_{1} - u_{2}|^{\frac{\delta}{2}} + |\mathbf{y}_{1} - \mathbf{y}_{2}|^{\delta} \le 2 (|u_{1} - u_{2}| + |\mathbf{y}_{1} - \mathbf{y}_{2}|^{2})^{\frac{\delta}{2}}.
\]
Due to $\big[ \tilde K(\hat w) \big]_{\sn}<\infty$, \eqref{U_term_bounded}, Assumption \ref{Assumption_1}, and (ii) of \lemref{lem:seminorm}, the map
$(t,l) \mapsto f_{x,y}(t,l)$ belongs to $C^{\frac\delta2,\delta}([0,T]\times\mathbb{R}) \cap C_{bc}([0,T]\times\mathbb{R})$ whose norm is bounded uniformly in $(x,y)$, and the same holds for the map $(t,l) \mapsto g(l,x,y)$, and $(t,l) \mapsto \tilde\mu(t,l,x,y)$, $(t,l) \mapsto \frac{1}{2} \sigma^{2}(t,e^{l})$ and $(t,l) \mapsto -\alpha-\lambda(t,e^{l})+\gamma(t,l,x,y)$.

Fix $(x,y) \in \mathbb{D}$ and perform the time reversal $s=T-t$. Define
\[
    a(s,l):=\frac{1}{2}\sigma^2(T-s,e^l), \qquad b_{x,y}(s,l):=\widetilde{\mu}(T-s,l,x,y),
\]
\[
    c_{x,y}(s,l):=-\alpha-\lambda(T-s,e^l)+\gamma(T-s,l,x,y), \qquad \check{f}_{x,y}(s,l):=-f_{x,y}(T-s,l).
\]
The time-reversed equation is the Cauchy problem for $u_{x, y}$
\[
    a(s,l) \dfrac{\partial^{2} u_{x, y}}{\partial l^{2}}(s,l) + b_{x,y}(s,l) \dfrac{\partial u_{x, y}}{\partial l}(s,l) + c_{x,y}(s,l) u_{x, y}(s,l) - \dfrac{\partial u_{x, y}}{\partial s}(s,l)=\check{f}_{x,y}(s,l),
\]
with initial condition $u_{x, y}(0,l)=e^{\alpha T}g(l,x,y)$.

The condition \eqref{Assumption_0} gives uniform parabolicity, while Lemma~\ref{lem:coeff}, \eqref{tilde_K_sn_finite}, and Lemma \ref{lem:seminorm} imply that the coefficients and $\check{f}_{x,y}$ are bounded and locally H\"older continuous in $l$, uniformly in $s$ and $(x,y)\in\mathbb{D}$. The initial datum is bounded and continuous. Hence, the hypotheses of \cite[Chapter~1, Theorem~12]{AvnerFriedman1964} are satisfied; in particular, the growth conditions in that theorem follow from boundedness. The representation given there, together with the Gaussian estimate for the fundamental solution, yields a bounded continuous solution $u_{x, y}$ of the time-reversed Cauchy problem.

Furthermore, Chapter~3, Theorem~10 of \cite{AvnerFriedman1964}, applied with $p=0$ on arbitrary bounded cylinders compactly contained in $(0,T] \times \mathbb{R}$, implies that $u_{x,y}$, $\frac{\partial u_{x, y}}{\partial s}$, $\frac{\partial u_{x, y}}{\partial l}$, $\frac{\partial^{2} u_{x, y}}{\partial l^{2}}$ are locally H\"older continuous. Transforming back to the original time variable and defining $\widetilde{w}(t,l,x,y):=u_{x,y}(T-t,l)$, we obtain
\[
    \widetilde{w}(\cdot,\cdot,x,y) \in C_{bc}([0,T]\times\mathbb{R}) \cap C_{\mathrm{loc}}^{1+\frac{\delta}{2},2+\delta}([0,T)\times\mathbb{R}), \qquad (x,y)\in\mathbb{D},
\]
and $\widetilde{w}$ solves \eqref{Linear_PDE}.

Next, we claim that any $\bar w(\cdot,\cdot,x,y)\in C^{1,2}([0,T)\times\mathbb{R})\cap C_{bc}([0,T]\times\mathbb{R})$ solving \eqref{Linear_PDE} coincides with the right-hand side of \eqref{contraction_map} evaluated at $\hat w$; that is, $\bar{w}(\cdot,\cdot,x,y)=\psi(\hat{w})(\cdot,\cdot,x,y)$ for each $(x, y) \in \D$. Indeed, fix $(t,\bx)$, let
\begin{equation*}
  \tau_{R}:=\inf\{s\ge t:|L^{(t,\bx)}_{s}-l|\ge R\}
\end{equation*}
and $\xi \in (0, T-t)$. Applying It\^o's lemma to $s\mapsto \mathcal{E}(s)\,\bar w(s,L^{(t,\bx)}_{s},x,y)$ on $[t,\tau_{R} \wedge (T-\xi)]$, where 
\[
  \mathcal{E}(s) := \exp \bigg( \int_{t}^{s} \Big[ -\alpha - \lambda \Big( a, e^{L_a^{(t,\mathbf{x})}} \Big) + \gamma(a, L_a^{(t,\mathbf{x})}, x, y) \Big]\, da \bigg),
\]
and then take expectations. Then the stochastic integral is a true martingale because on the stopped interval the process remains in the compact set $[t,T-\xi]\times[l-R,l+R]$, where $\frac{\partial \bar{w}}{\partial l}$ is bounded, and $\mathcal{E}$ is bounded by (i) of \lemref{lem:coeff}. By \eqref{Linear_PDE}, the time integral term is $-\mathcal{E}(s)f_{x,y}(s,L_s^{(t,\mathbf{x})})\,ds$. Moving its integral to the other side and taking expectations gives the following identity.
\begin{align*}
  \bar w(t,\bx) = \mathbb{E} \Bigg[ \mathcal{E} \big(\tau_{R} \wedge (T-\xi)\big)\,
\bar{w} \Big( \tau_{R} \wedge (T-\xi), L^{(t,\bx)}_{\tau_{R} \wedge (T-\xi)},x,y \Big) + \int_{t}^{\tau_{R} \wedge (T-\xi)} \mathcal{E}(s) f_{x,y}(s,L_s^{(t,\mathbf{x})}) \,ds
\Bigg].
\end{align*}
Since $\sup_{s \in [t,T]}|L^{(t,\bx)}_{s}-l|<\infty$ a.s. due to (i) of \lemref{lem:random}, $\tau_{R} \wedge (T-\xi) \uparrow T-\xi$ as $R \to \infty$, and $\bar w$, $f_{x,y}$, $g$, and $\mathcal{E}$ are bounded, so the dominated convergence can be applied; letting then $\xi \downarrow 0$ and using the continuity of $\bar{w}$ up to $t=T$ together with the terminal condition in \eqref{Linear_PDE}, we obtain $\bar w(t,\bx)=\psi(\hat w)(t,\bx)$, which proves the
claim.

In particular, $\tilde{w}(\cdot,\cdot,x,y)=\psi(\hat{w})(\cdot,\cdot,x,y)=\hat{w}(\cdot,\cdot,x,y)$ for each $(x, y) \in \D$; hence $\tilde{w}=\hat{w}$ on $[0,T] \times \mathbb{R} \times \D$. Since we know $\hat{w} \in C_{bc}([0,T]\times\mathbb{R}\times \D)$, $\nu:=\hat{w}$ is the unique classical solution of \eqref{eq:HJB_nu}. Finally, let $u$ be any other bounded classical solution of \eqref{eq:HJB_nu}. Applying the same Feynman-Kac argument to $u$ shows that $u=\psi(u)$. Since $\psi$ has a unique fixed point in $C_{bc}([0,T]\times\mathbb{R}\times\mathbb{D})$, it follows that $u=\hat{w}$. This proves \eqref{nu_regularity_1}.

It remains to prove \eqref{nu_regularity_2}. Work in the reversed time variable and write
\[
    u_{x,y}(s,l):=\nu(T-s,l,x,y), \qquad \check{f}_{x,y}(s,l):=-f_{x,y}(T-s,l).
\]
For $\xi\in(0,T)$ and $R>0$, define the nested cylinders
\[
    Q_{\xi,R} := [\xi,T] \times [-R,R], \qquad
    \widetilde{Q}_{\xi,R} := \Big[ \frac{\xi}{2},T \Big] \times [-R-1,R+1].
\]
After translating the time variable by $\xi/2$, Chapter~3, Theorem~5 of \cite{AvnerFriedman1964} applies on the larger cylinder $\widetilde{Q}_{\xi,R}$; the norm appearing in that theorem is the weighted interior norm defined in Chapter~3, Section~2. Since $Q_{\xi,R}$ has a strictly positive parabolic distance from the initial and lateral boundary of $\widetilde{Q}_{\xi,R}$, the weighted estimate implies the ordinary interior estimate
\[
    \| u_{x,y}\|_{C^{1+\frac{\delta}{2},2+\delta}(Q_{\xi,R})}
    \le C_{\xi,R} \Big( \|u_{x,y}\|_{\infty} + \|\check{f}_{x,y}\|_{C^{\frac{\delta}{2},\delta}(\widetilde{Q}_{\xi,R})} \Big),
\]
where $C_{\xi,R}$ is independent of $(x,y)\in\mathbb{D}$. Returning to the original time variable gives
\[
    \| \nu(\cdot,\cdot,x,y) \|_{C^{1+\frac{\delta}{2},2+\delta}([0,T-\xi] \times [-R,R])} \le C_{\xi,R} \Big( \|\nu\|_{\infty} + \| f_{x,y} \|_{C^{\frac{\delta}{2},\delta}([0,T-\xi/2] \times [-R-1,R+1])} \Big).
\]
By \eqref{w_n_boundedness}, \eqref{tilde_K_sn_finite}, Lemma~\ref{lem:seminorm}, and the boundedness and H\"older continuity of $\lambda$, the right-hand side is bounded uniformly in $(x,y)\in\mathbb{D}$. Taking the supremum over $(x,y)$ proves \eqref{nu_regularity_2}.
\end{proof}

On the original filtered probability space, fix an admissible strategy up to time $t$, and let $(W_t^{(0)},H_t)$ be the resulting portfolio state. Let $\mathcal{N}_t$ denote the class of admissible continuation controls on $(t,T]$, all initialized from this same state, and define the conditional value by
\begin{equation}\label{eq:conditional_value}
    \mathcal{V}_t := \operatorname*{ess\,sup}_{N \in \mathcal{N}_t}
    \mathbb{E} \Big[ U\big(W_T^{\mathrm{liq},N}\big)\,\Big|\,\mathcal{F}_t \Big], \qquad t\in[0,T]
\end{equation}
where $W_T^{\mathrm{liq},N}$ denotes the terminal liquidated wealth under the admissible strategy $N$. The following theorem identifies the value function and verifies the optimal feedback control.

\begin{theorem}[Verification] \label{Verification}
    Let $\nu$ be the unique bounded classical solution of \eqref{eq:HJB_nu}, and define
    \[
        \overline{V}(t,l,x,y) := e^{-\alpha t} \nu(t,l,x,y) \cU(x+ye^l).
    \]
    Then
    \begin{equation*}
        V(t,l,x,y) = \overline{V}(t,l,x,y), \qquad (t,l,x,y)\in[0,T]\times\mathbb{R}\times\mathbb{D}
    \end{equation*}
    where $V$ is defined in \eqref{Value_define}. Moreover, the conditional value process defined in \eqref{eq:conditional_value} satisfies for every fixed $t\in[0,T]$
    \begin{equation}\label{eq:conditional_verification}
        \mathcal{V}_t = V \big(t,L_t,W_t^{(0)},H_t \big) = \overline{V} \big( t,L_t,W_t^{(0)},H_t \big) \quad \text{a.s.}
    \end{equation}
    The measurable feedback selector constructed in the proof is optimal.
\end{theorem}

\begin{proof}
For $t=T$, both the deterministic and conditional assertions follow directly from the terminal condition. Hence, fix $t\in[0,T)$.

Since $W^{(0)}$ and $H$ are constant between consecutive jump times of $P$, we consider 
\begin{equation} \label{Ito_target}
   e^{- \alpha s} \nu \big( s, L_{s}, W_{s}^{(0)}, H_{s} \big) \cU \big( W_{s}^{(0)} + H_{s} e^{L_{s}} \big)
\end{equation}
on each inter-jump time interval first. At a jump time $u$, the jump part is
\begin{align}
    & e^{-\alpha u} \nu \big( u,L_u,W_u^{(0)},H_u \big) \cU \big( W_u^{(0)}+H_u e^{L_u} \big) - e^{-\alpha u} \nu \big( u,L_u,W_{u-}^{(0)},H_{u-} \big) \cU\big( W_{u-}^{(0)}+H_{u-}e^{L_u} \big) \nonumber
    \\
    &= e^{-\alpha u} \cU \big( W_{u-}^{(0)}+H_{u-}e^{L_u} \big) \Big[ K(\nu) \big( u,L_u,W_{u-}^{(0)},H_{u-},Z_u \big) - \nu \big( u,L_u,W_{u-}^{(0)},H_{u-} \big) \Big]. \label{Jump_term}
\end{align}
Therefore we get the cumulative jump part of \eqref{Ito_target} on $s \in \big[ t, \tau_{R}^{L} \wedge (T-\xi) \big]$ for given $\xi \in (0,T-t)$, where $\tau_{R}^{L}$ is defined as
\[
    \tau_R^{L} := \inf\{s \ge t : |L_{s}| \ge R\},
\]
by summing \eqref{Jump_term} as follows: 
\begin{equation} \label{Cumulative_jump_part}
    \int_{t}^{\tau_{R}^{L} \wedge (T-\xi)} e^{-\alpha s} \cU \big( W_{s-}^{(0)}+H_{s-}e^{L_s} \big) \Big[ K(\nu) \big( s,L_s,W_{s-}^{(0)},H_{s-},Z_s \big) - \nu \big( s,L_s,W_{s-}^{(0)},H_{s-} \big) \Big]\,dP_s.
\end{equation}
Using 
\[
    \tilde{P}_s := P_s-\int_0^s \lambda \big( u,e^{L_u} \big)\,du, \qquad dP_s = d\tilde{P}_s+\lambda \big( s,e^{L_s} \big)\,ds,
\]
where $\tilde{P}$ is a martingale, \eqref{Cumulative_jump_part} is divided into the compensator part
\[
    \int_{t}^{\tau_{R}^{L} \wedge (T-\xi)} e^{-\alpha s} \cU \big( W_{s-}^{(0)}+H_{s-}e^{L_s} \big) \lambda \big( s,e^{L_s} \big) \Big[ K(\nu) \big( s,L_s,W_{s-}^{(0)},H_{s-},Z_s \big) - \nu \big( s,L_s,W_{s-}^{(0)},H_{s-} \big) \Big]\,ds.
\]
and local martingale part 
\[
    \int_{t}^{\tau_{R}^{L} \wedge (T-\xi)} e^{-\alpha s} \cU \big( W_{s-}^{(0)}+H_{s-}e^{L_s} \big) \Big[ K(\nu) \big( s,L_s,W_{s-}^{(0)},H_{s-},Z_s \big) - \nu \big( s,L_s,W_{s-}^{(0)},H_{s-} \big) \Big]\,d \tilde{P}_{s}.
\]
Hence, we apply It\^o's lemma for semimartingale (see, e.g., Proposition 8.19 of \cite{RamaCont2004}) in the variables $(s,L_{s})$ on $s \in \big[ t, \tau_{R}^{L} \wedge (T-\xi) \big]$ to \eqref{Ito_target} so that 
\begin{align} \label{value_Ito}
  & e^{- \alpha s} \nu \big( s, L_{s}, W_{s}^{(0)}, H_{s} \big) \cU \big( W_{s}^{(0)} + H_{s} e^{L_{s}} \big) \big|_{s = \tau_{R}^{L} \wedge (T-\xi)} - e^{- \alpha t} \nu \big( t, L_{t}, W_{t}^{(0)}, H_{t} \big) \cU \big( W_{t}^{(0)} + H_{t} e^{L_{t}} \big) \\
  & \quad = \int_{t}^{\tau_{R}^{L} \wedge (T-\xi)} e^{- \alpha s} \cU \big( x + y e^{l} \big) \bigg[ \frac{\partial \nu}{\partial s}(s,l,x,y) + \widetilde{\mu}(s,l,x,y) \frac{\partial \nu}{\partial l}(s,l,x,y) + \frac{\sigma^{2} \big( s, e^{l} \big)}{2} \frac{\partial^{2} \nu}{\partial l^{2}}(s,l,x,y) \nonumber \\
  & \qquad + \Big( - \alpha - \lambda \big( s,e^{l} \big) + \gamma(s,l,x,y) \Big) \nu(s,l,x,y) + \lambda \big( s,e^{l} \big) K(\nu) (s,l,x,y,Z_{s}) \bigg] \bigg|_{(l,x,y) = (L_{s-}, W_{s-}^{(0)}, H_{s-})} d s \nonumber \\
  & \quad \; + \int_{t}^{\tau_{R}^{L} \wedge (T-\xi)} e^{- \alpha s} \cU \big( x + y e^{l} \big) \sigma \big( s, e^{l} \big) \bigg( \frac{\partial \nu}{\partial l}(s,l,x,y) \nonumber \\
  & \qquad \quad + \frac{\nu(s,l,x,y) y e^{l} \cU' \big( x + y e^{l} \big)}{\cU \big( x + y e^{l} \big)} \bigg) \bigg|_{(l,x,y) = (L_{s-}, W_{s-}^{(0)}, H_{s-})} d B_{s} \nonumber \\
  & \quad \; + \int_{t}^{\tau_{R}^{L} \wedge (T-\xi)} e^{-\alpha s} \cU \big( x + y e^{l} \big) \Big[ K(\nu) (s,l,x,y,Z_s) - \nu(s,l,x,y) \Big] \bigg|_{(l,x,y) = (L_{s-}, W_{s-}^{(0)}, H_{s-})} \,d \tilde{P}_{s} \nonumber
\end{align}
where 
\begin{equation*}
    Z_{s} := \frac{N_{s}+H_{s-}}{\frac{W_{s-}^{(0)} e^{- L_{s}}}{1 + \overline{\epsilon}(e^{L_{s}})} + H_{s-}}.
\end{equation*}
On the full-probability event in Definition \ref{def:admissible},
\[
    N_s\in\mathcal{M} \big( L_s,W_{s-}^{(0)},H_{s-} \big) \quad \text{and} \quad Z_s\in[0,1]
\]
for every $s\in(0,T]$, so that all integrals in \eqref{value_Ito} are well defined.

Next, we check that the integrals with respect to $d B_{s}$ and $d \tilde{P}_{s}$ are martingales. For the integral with respect to $d B_{s}$, using Notation~\ref{not:master} together with \eqref{eq:Fbounds}, we have 
\[
  \bigg| \frac{H_{s-} e^{L_{s}}}{W_{s-}} \bigg| = \pi_{s-} \in [0,1] \qquad \text{and} \qquad \bigg| \frac{W_{s-} \cU'(W_{s-})}{\cU(W_{s-})} \bigg| = | F'(\ln W_{s-}) | \le \| F' \|_{\infty} < \infty,
\]
so the integrand of $d B_{s}$ is bounded by 
\[
  \cU(W_{s-}) \| \sigma \|_{\infty} \bigg( \sup_{(x, y) \in \D} \| \nu(\cdot, \cdot, x, y) \|_{C^{1+\frac{\delta}{2},2+\delta} ([0, T-\xi] \times [-R,R])} + \| \nu \|_{\infty} \| F' \|_{\infty} \bigg).
\]
Hence, the corresponding stochastic integral is a square-integrable martingale due to \eqref{nu_regularity_1}, \eqref{nu_regularity_2}, and \eqref{Square_integrable}. For the integral with respect to $d \tilde{P}_{s}$, we know $|K(\nu) - \nu| \le (c_{\cU}+1) \|\nu\|_{\infty}$ so the integrand of $d \tilde{P}_{s}$ is bounded by $(c_{\cU}+1) \|\nu\|_{\infty} \cU(W_{s-})$, hence the integral is a square-integrable martingale due to \eqref{nu_regularity_1} and \eqref{Square_integrable}. Taking expectations under $\mathbb{E}_{t,l,x,y}$ in \eqref{value_Ito}, and using \eqref{eq:HJB_nu}, the expression in square brackets in the $ds$-integral of \eqref{value_Ito} equals
\[
    \lambda \big( s,e^{L_s} \big) \Big[ K(\nu) \big( s,L_s,W_{s-}^{(0)},H_{s-},Z_s \big) - \sup_{z\in[0,1]} K(\nu) \big( s,L_s,W_{s-}^{(0)},H_{s-},z \big) \Big] \le 0.
\]
Therefore, we have 
\begin{equation*}
  \mathbb{E}_{t,l,x,y} \Big[ e^{- \alpha s} \nu \big( s, L_{s}, W_{s}^{(0)}, H_{s} \big) \cU \big( W_{s}^{(0)} + H_{s} e^{L_{s}} \big) |_{s = \tau_{R}^{L} \wedge (T-\xi)} \Big] \le e^{- \alpha t} \nu(t, l, x, y) \cU \big( x + y e^{l} \big).
\end{equation*}
In the following, we claim that
\begin{align} 
    \mathbb{E}_{t,l,x,y} & \bigg[ e^{- \alpha T} \nu \big( T, L_{T}, W_{T}^{(0)}, H_{T} \big) \cU \big( W_{T}^{(0)} + H_{T} e^{L_{T}} \big) \bigg] \label{Double_limits_Expectation}
    \\
    & = \lim_{\xi \downarrow 0} \lim_{R \to \infty} \mathbb{E}_{t,l,x,y} \bigg[ e^{- \alpha s} \nu \big( s, L_{s}, W_{s}^{(0)}, H_{s} \big) \cU \big( W_{s}^{(0)} + H_{s} e^{L_{s}} \big) \Big|_{s = \tau_{R}^{L} \wedge (T-\xi)} \bigg].  \nonumber
\end{align}
For notational convenience, define
\[
    \Theta_s := e^{-\alpha s} \nu \big( s,L_s,W_s^{(0)},H_s \big) \cU \big( W_s^{(0)}+H_s e^{L_s} \big), \qquad \tau_{R,\xi} := \tau_R^L \wedge (T-\xi).
\]
By \eqref{nu_regularity_1} and Lemma \ref{lem:uniform-integrability},
\begin{align*}
    \sup_{(R,\xi) \in (0, \infty) \times (0,T-t)} \mathbb{E}_{t,l,x,y} \Big[ \big| \Theta_{\tau_{R,\xi}} \big|^{2} \Big]
    &\le \|\nu\|_\infty^2 \, \mathbb{E}_{t,l,x,y} \bigg[ \sup_{u\in[t,T]} \cU^2(W_u) \bigg] < \infty.
\end{align*}
Hence the family $\{\Theta_{\tau_{R,\xi}} : (R,\xi) \in (0, \infty) \times (0,T-t)\}$ is uniformly integrable.

Fix $\xi\in(0,T-t)$. Since $L$ has continuous paths on the compact interval $[t,T-\xi]$, we have $\tau_R^L>T-\xi$ for all sufficiently large $R$, almost surely. Consequently, $\Theta_{\tau_{R,\xi}} \to \Theta_{T-\xi}$ as $R\to\infty$ a.s. Uniform integrability and Vitali's convergence theorem therefore imply
\[
    \lim_{R\to\infty} \mathbb{E}_{t,l,x,y} [\Theta_{\tau_{R,\xi}}] = \mathbb{E}_{t,l,x,y} [\Theta_{T-\xi}].
\]

Next, since the compensator of $P$ is continuous,
\[
    \mathbb{E}[\Delta P_T] = \mathbb{E} \bigg[ \int_{\{T\}} \lambda(s,e^{L_s})\,ds \bigg] = 0.
\]
As $\Delta P_T\ge0$, it follows that $\Delta P_T=0$, hence $W_{T-}^{(0)}=W_T^{(0)}$ and $H_{T-}=H_T$ a.s. Using the continuity of $L$ and $\nu$, we obtain $\Theta_{T-\xi} \to \Theta_T$ as $\xi \to 0$ a.s. The family $\{\Theta_{T-\xi} : \xi \in (0,T-t)\}$ is uniformly integrable by the same $L^2$ estimate above. Another application of Vitali's convergence theorem gives
\[
    \lim_{\xi \downarrow 0} \mathbb{E}_{t,l,x,y} [\Theta_{T-\xi}] = \mathbb{E}_{t,l,x,y} [\Theta_T].
\]
Combining the two limits proves \eqref{Double_limits_Expectation}.

Finally, combine \eqref{Double_limits_Expectation} with the terminal condition of \eqref{eq:HJB_nu} to derive 
\begin{equation*}
    \mathbb{E}_{t,l,x,y} \big[ U \big(W_T^{\mathrm{liq},N}\big) \big] \le e^{- \alpha t} \nu(t, l, x, y) \cU \big( x + y e^{l} \big).
\end{equation*}
Fix $(t,l,x,y) \in [0,T] \times \R \times \D$. To make the above inequality an equality, we can choose a Borel measurable $z^{*}(s,l,x,y)$ such that 
\begin{equation*}
  z^{*}(s,l,x,y) \in \argmax_{z \in [0,1]} K(\nu)(s, l, x, y, z)
\end{equation*}
because $(s,l,x,y,z) \mapsto K(\nu)(s,l,x,y,z)$ is continuous on $[0,T] \times \R \times \D \times [0,1]$ (e.g. \lemref{lem:measurable}).

For $s \in [t,T]$, let $\hat{W}_{s}^{(0),t,l,x,y}, \hat{H}_{s}^{\,t,l,x,y}$ be the solutions of SDEs in \eqref{SDE_W0} and \eqref{SDE_H} starting at 
\[
    \hat{W}_t^{(0),t,l,x,y}:=x,\qquad \hat{H}_t^{t,l,x,y}:=y
\]
respectively, equipped with the following trading strategy 
\[
  \hat{N}_{s}^{t,l,x,y}:=
    \begin{cases}
    0, & s = t,\\ \theta \big( L_{s}, \hat{W}_{s-}^{(0),t,l,x,y}, \hat{H}_{s-}^{\,t,l,x,y}, \hat{Z}_{s}^{\,t,l,x,y} \big), & s \in (t,T]
    \end{cases}
\]
recursively at the jump times of $P$ (which is possible because $P$ has finitely many jumps on $[0,T]$ almost surely). Here,  
\[
  \hat{Z}_{s}^{t,l,x,y}: = z^* \big( s, L_{s}, \hat{W}_{s-}^{(0),t,l,x,y}, \hat{H}_{s-}^{t,l,x,y} \big).
\]
Such solutions exist and are uniquely determined due to the structures of \eqref{SDE_W0} and \eqref{SDE_H}.

We can also check that \\
(1) $s \mapsto L_{s}$ is continuous and adapted, thus predictable; \\
(2) $s \mapsto \hat{W}_{s-}^{(0),t,l,x,y}$ and $s \mapsto \hat{H}_{s-}^{t,l,x,y}$ are predictable, hence $s \mapsto \hat{N}_{s}^{t,l,x,y}$ is predictable; \\
(3) Since $\hat{Z}_{s}^{\,t,l,x,y} \in [0,1]$, the definition of $\theta$ implies 
\[
    \hat{N}_s^{t,l,x,y} \in \bigg[ -\hat{H}_{s-}^{t,l,x,y}, \frac{\hat{W}_{s-}^{(0),t,l,x,y}}{[1+\overline{\epsilon}(e^{L_s})]e^{L_s}} \bigg], \qquad s\in(t,T].
\]
Together with $\hat{N}_t^{t,l,x,y}=0$, this proves that $\hat{N}^{t,l,x,y}\in\mathcal{N}(t,l,x,y)$. Now we let $\hat{W}_{s}^{liq} := \hat{W}_{s}^{(0),t,l,x,y} + [1 - \underline{\epsilon}(e^{L_{s}})] \hat{H}_{s}^{t,l,x,y} e^{L_{s}}$. Using these processes, we conclude 
\begin{equation*}
  V(t,l,x,y) = \sup_{N \in \mathcal{N}(t,l,x,y)} \mathbb{E}_{t,l,x,y} \Big[ U \big(W_T^{\mathrm{liq},N} \big) \Big] = \mathbb{E}_{t,l,x,y} \Big[ U \big( W_T^{\mathrm{liq},\hat{N}^{t,l,x,y}} \big) \Big] = \overline{V}(t,l,x,y).
\end{equation*}
We finally return to the original filtered probability space. For $t=T$, \eqref{eq:conditional_verification} follows directly from the terminal condition. Hence, fix $t\in[0,T)$. Let $N\in\mathcal{N}_t$ be arbitrary and, for $R>0$ and $\xi\in(0,T-t)$, set $\tau_{R,\xi}:=\tau_R^L\wedge(T-\xi)$. Repeating the argument of stopped It\^o's lemma on $[t,\tau_{R,\xi}]$ and taking conditional expectations with respect to $\mathcal{F}_t$, we obtain almost surely
\[
    \mathbb{E} \Big[ e^{-\alpha\tau_{R,\xi}} \nu \big( \tau_{R,\xi},L_{\tau_{R,\xi}},W_{\tau_{R,\xi}}^{(0),N},H_{\tau_{R,\xi}}^N \big) \cU \big( W_{\tau_{R,\xi}}^{(0),N}+H_{\tau_{R,\xi}}^Ne^{L_{\tau_{R,\xi}}} \big) \Big| \mathcal{F}_t \Big]
    \le \overline{V}\big(t,L_t,W_t^{(0)},H_t\big).
\]
For the original-space argument, concatenate $N\in\mathcal{N}_t$ with the admissible strategy used up to time $t$. The original initial state estimate in Lemma~\ref{lem:uniform-integrability} then implies that the stopped family is uniformly integrable. Hence, by the same terminal limit argument used after \eqref{Double_limits_Expectation}, the random variable inside the conditional expectation converges in $L^1$ to $U(W_T^{\mathrm{liq},N})$, first as $R\to\infty$ and then as $\xi\downarrow0$. Since conditional expectation is a contraction on $L^1$, the corresponding conditional expectations converge in $L^1$. Passing to an almost surely convergent subsequence and using the stopped inequality, we obtain
\[
    \mathbb{E} \Big[ U \big( W_T^{\mathrm{liq},N} \big) \, \Big| \, \mathcal{F}_t \Big] \le \overline{V} \big( t,L_t,W_t^{(0)},H_t \big),\qquad \text{a.s.}
\]

Let $z^*$ be the Borel measurable maximizing selector constructed above. Initialize the feedback-controlled state process by $\hat{W}_t^{(0),t}:=W_t^{(0)}$, $\hat{H}_t^{t}:=H_t$. For $s\in(t,T]$, define
\[
    \hat{Z}_s^{t}:=z^* \big( s,L_s,\hat{W}_{s-}^{(0),t},\hat{H}_{s-}^{t} \big), \qquad \hat{N}_s^{t}:=\theta \big( L_s,\hat{W}_{s-}^{(0),t},\hat{H}_{s-}^{t},\hat{Z}_s^{t} \big),
\]
and set $\hat{N}_t^{t}:=0$. Since $P$ has finitely many jumps on $[t,T]$ almost surely, the corresponding state process $(\hat{W}^{(0),t},\hat{H}^{t})$ is constructed recursively at the jump times of $P$.

The process $L$ is continuous and adapted, and hence predictable, while $\hat{W}_{-}^{(0),t}$ and $\hat{H}_{-}^{t}$ are predictable. Therefore, the Borel measurability of $z^*$ implies that $\hat{Z}^{t}$ and $\hat{N}^{t}$ are predictable. Since $\hat{Z}_s^{t}\in[0,1]$ and $\theta(l,x,y,\cdot)$ maps $[0,1]$ into $\mathcal{M}(l,x,y)$, we conclude that $\hat{N}^{t}\in\mathcal{N}_t$. 

For this feedback control, the finite-variation drift vanishes. Repeating the preceding stopped argument with equality and using the same $L^1$ passage to the terminal time yields
\[
    \mathbb{E} \Big[ U \big( W_T^{\mathrm{liq},\hat{N}^{t}} \big) \Big| \mathcal{F}_t \Big] = \overline{V} \big( t,L_t,W_t^{(0)},H_t \big), \qquad \text{a.s.}
\]
Taking the essential supremum over $N\in\mathcal{N}_t$, we conclude that
\[
    \mathcal{V}_t = \overline{V} \big( t,L_t,W_t^{(0)},H_t \big) = V \big( t,L_t,W_t^{(0)},H_t \big), \qquad \text{a.s.}
\]
This proves \eqref{eq:conditional_verification}.
\end{proof}

\bigskip

\section{Optimal portfolio} \label{sec:optimal-portfolio}

Although the utility function $U$ may not be concave, we can characterize an optimal trading strategy.

\begin{corollary}[Optimal portfolio] \label{Optimal_portfolio}
    Given current time $t \in [0, T)$, current log price $l \in \R$, and current positions $(x,y) \in \mathbb{D}$, define the set of optimal trade parameters 
    \[
      \hat{Z}(t, l, x, y) := \argmax_{z \in [0,1]} K(\nu)(t,l,x,y,z)
    \]
    and the two sets of maximizers 
    \begin{align*}
      R_{\mathrm{buy}} & := \argmax_{z \in [b(l,x,y),1]} \bigg( \nu \big( t,l, x - [1+\overline\epsilon(e^{l})] e^{l} \theta(l,x,y,z), y + \theta(l,x,y,z) \big) \frac{\cU \big( x + y e^{l} - \overline{\epsilon}(e^{l}) e^{l} \theta(l,x,y,z) \big)}{\cU \big( x + y e^{l} \big)} \bigg), \\
      R_{\mathrm{sell}} & := \argmax_{z \in [0,b(l,x,y)]} \bigg( \nu \big( t,l, x - [1-\underline\epsilon(e^{l})] e^{l} \theta(l,x,y,z), y + \theta(l,x,y,z) \big) \frac{\cU \big( x + y e^{l} + \underline{\epsilon}(e^{l}) e^{l} \theta(l,x,y,z) \big)}{\cU \big( x + y e^{l} \big)} \bigg).
    \end{align*}
    Then $\hat{Z}(t,l,x,y)$ is characterized as 
    \begin{equation*}
      \hat{Z}(t,l,x,y) = \argmax_{z \in R_{\mathrm{buy}} \cup R_{\mathrm{sell}}} K(\nu)(t,l,x,y,z).
    \end{equation*}
\end{corollary}

\begin{proof}
    We observe that $R_{\mathrm{buy}}$ and $R_{\mathrm{sell}}$ are well defined due to \lemref{lem:value_regularity}. We observe that for $(z_{\mathrm{buy}}, z_{\mathrm{sell}}) \in R_{\mathrm{buy}} \times R_{\mathrm{sell}}$, 
    \begin{equation*}
        \max_{z \in [0,1]} K(\nu)(t,l,x,y,z) = \max \Big\{ K(\nu)(t,l,x,y,z_{\mathrm{buy}}), K(\nu)(t,l,x,y,z_{\mathrm{sell}}) \Big\}.
    \end{equation*}
\end{proof}

\begin{remark}[Interpretation of an optimal trade]
    For an optimal selector $\hat{z}$, the corresponding trade is a purchase, a sale, or no trade according as $\hat{z}>b(l,x,y)$, $\hat{z}<b(l,x,y)$, or $\hat{z}=b(l,x,y)$. The optimal selector is chosen as the case among the three above that yields the largest post-trade value $K(\nu)(t,l,x,y,\hat{z})$, so that an investor maximizes \eqref{Objective}. Since these argmax sets need not be singletons for a general utility $U$ (for an extreme example, $U$ is constant), the optimal selector may not be unique. The numerical examples in the next section illustrate that, for nonconcave utilities and at fixed time and pre-trade wealth, the computed buying, selling, and no-trade sets need not be single intervals or follow the usual ordering. They also exhibit abrupt switches in the optimal post-trade risky weight.
\end{remark}

\bigskip

\section{Numerical examples}    \label{sec:numerical}

To observe unusual phenomena in optimal trading, we consider the following utility functions 
\begin{align*}
  U_{1}(w) & = - \frac{1}{w} + 6 \, \Phi \Big( \frac{w-1.2}{0.1} \Big) + 3.5 \, \Phi \Big( \frac{w-1.8}{0.2} \Big), \\
  U_{2}(w) & = - \frac{1}{w^{2}} + 8 \, \Phi \Big( \frac{w-1.8}{0.15} \Big), \\
  U_{3}(w) & = - \frac{1}{w} + 6 \, \Phi \Big( \frac{w-1.6}{0.02} \Big) + 12 \, \Phi \Big( \frac{w-2.1}{0.05} \Big)
\end{align*}
where $\Phi$ is the cumulative distribution function of the standard normal distribution. It is straightforward to check that each $U_i$ is strictly increasing but neither concave nor convex in $w$. All three utilities satisfy Assumption \ref{U_assumption} with $p=2$ and $\delta_U=1$. Indeed,
\[
    \max_{i\in\{1,2,3\}}\sup_{w>0}\frac{|U_i(w)|+w|U_i'(w)|}{\mathcal{U}_2(w)}<\infty,
\]
which implies that $x\mapsto U_i(e^x)/\mathcal{U}_2(e^x)$ is bounded and Lipschitz continuous.

\begin{figure}[ht]
\centering
\includegraphics[width=.7\columnwidth]{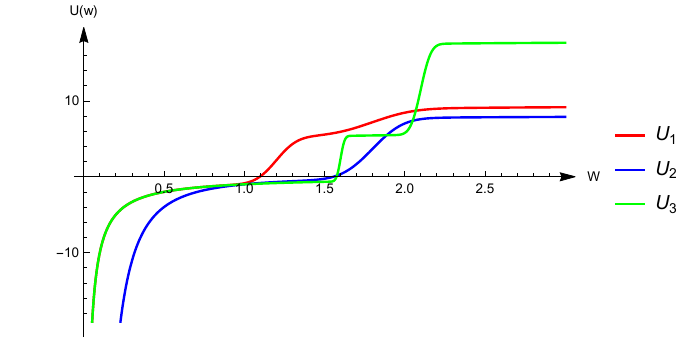}
\caption{The figure shows the form of the utility functions $U_{1}, U_{2}, U_{3}$ as functions of total wealth, used in the numerical examples.} \label{Utility_function}
\end{figure}

Throughout this section, $\mu$, $\sigma$, $\lambda$, $\overline{\epsilon}$, and $\underline{\epsilon}$ are constants:
\[
    \mu = 0.06, \quad \sigma = 0.35, \quad \lambda = 20, \quad T = 1,
\]
and the value of $\overline{\epsilon} = \underline{\epsilon}$ will be specified later. The dependence of $\nu(t,l,x,y)$ on $(l,y)$ then enters only through the risky dollar position $m:=ye^l$, so we write $\nu(t,l,x,y)=\widetilde{\nu}(t,m,x)$.

For $(m,x)\in\mathbb{D}$ and $z\in[0,1]$, define the post-trade risky position and cash position by
\[
    m_z:=z\Big(\frac{x}{1+\overline{\epsilon}}+m\Big), \qquad
    x_z:=x-(1+\overline{\epsilon})(m_z-m)^++(1-\underline{\epsilon})(m_z-m)^-.
\]
Here, $m_z-m=e^l\theta(l,x,y,z)$ is the dollar amount traded. Thus $m_z>m$, $m_z<m$, and $m_z=m$ correspond to a purchase, a sale, and no trade, respectively. For a test function $\phi$, define the reduced normalized intervention operator by
\[
    \overline{K}(\phi)(t,m,x,z):=\phi(t,m_z,x_z)\frac{\cU(x_z+m_z)}{\cU(x+m)}.
\]
The coefficients $\widetilde{\mu}$ and $\gamma$ in the equation below are understood after substituting $ye^l=m$; they generally remain dependent on $(m,x)$. The reduced function $\widetilde{\nu}$ satisfies
\begin{equation*}
    \begin{cases}
    0 = \dfrac{\partial \tilde{\nu}}{\partial t}(t,m,x) + \left( \widetilde{\mu} +  \dfrac{\sigma^{2}}{2} \right) m \dfrac{\partial \tilde{\nu}}{\partial m}(t,m,x) + \dfrac{\sigma^{2} m^{2}}{2} \dfrac{\partial^{2}\tilde{\nu}}{\partial m^{2}}(t,m,x) + \big( -\alpha - \lambda + \gamma \big) \tilde{\nu}(t,m,x)
    \\[0.8em]
    \qquad + \lambda\displaystyle\sup_{z\in[0,1]}\overline{K}(\widetilde{\nu})(t,m,x,z), \hfill (t,m,x)\in[0,T)\times (0, \infty) \times [0, \infty),
    \\[0.8em]
    0 = \dfrac{\partial \tilde{\nu}}{\partial t}(t,0,x) - (\alpha + \lambda) \widetilde{\nu}(t,0,x)+\lambda\displaystyle\sup_{z\in[0,1]}\overline{K}(\widetilde{\nu})(t,0,x,z), \hfill (t,x) \in [0, T) \times (0, \infty),
    \\[0.8em]
    \tilde{\nu}(T,m,x) = e^{\alpha T} \tilde{g}(m,x), \hfill (m,x) \in \mathbb{D},
    \end{cases}
\end{equation*}
where 
\[
  \tilde{g}(m,x) := \frac{U ( x+[1-\underline{\epsilon}] m )}{\cU (x + m)}.
\]
The code is implemented by the finite difference method (FDM) with an Implicit-Explicit Method (IMEX) using Python, where the spatial domain $(m,x)$ is uniformly partitioned into 2000 subintervals and the time interval $t$ into 4000 subintervals.

For each fixed pre-trade total wealth $w>0$, we vary the pre-trade risky weight $\pi\in[0,1]$ by setting
\[
    m=w\pi,\qquad x=w(1-\pi).
\]
For a maximizing control parameter $\hat{z}$, the corresponding post-trade risky weight is
\[
    \hat{\pi}:=\frac{m_{\hat{z}}}{x_{\hat{z}}+m_{\hat{z}}}, \qquad
    \hat{z}\in\argmax_{z\in[0,1]}\overline{K}(\widetilde{\nu})(t,m,x,z).
\]
When the optimizer is not unique, $\hat{\pi}$ denotes the post-trade risky weight corresponding to the maximizer returned by the numerical procedure.

The figures plot a numerical approximation of $\hat{\pi}$ as a function of $\pi$ at fixed pre-trade wealth $w$. A purchase, a sale, and no trade correspond respectively to $\hat{\pi}>\pi$, $\hat{\pi}<\pi$, and $\hat{\pi}=\pi$. We evaluate the optimal rebalancing rule at $t=0$, interpreting it as the portfolio adjustment that would be chosen if a trading opportunity were available at that state. This does not represent an executed trade at time $0$, which is excluded by the model. These plots show one-dimensional slices at fixed time and pre-trade wealth; they do not by themselves establish disconnectedness of the corresponding regions in the full state space.

\begin{figure}[ht]
\centering
\includegraphics[width=.45\columnwidth]{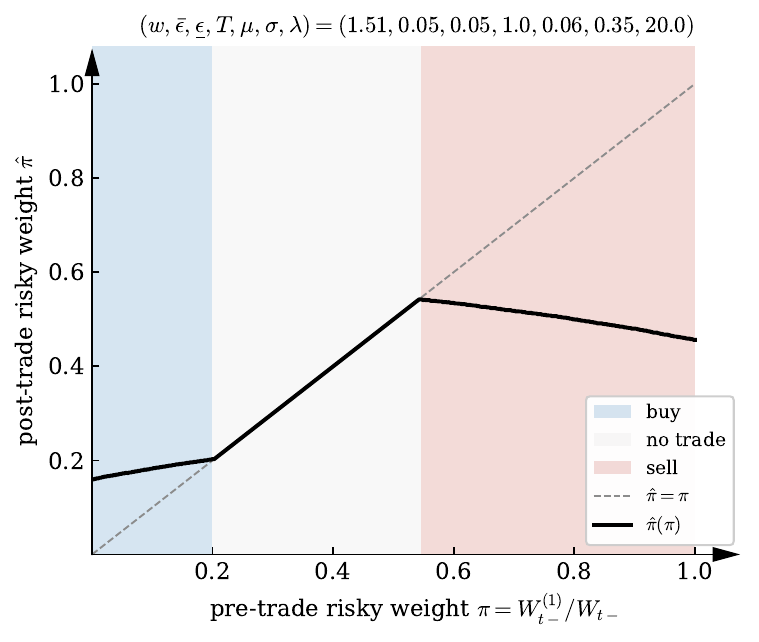} \quad \includegraphics[width=.45\columnwidth]{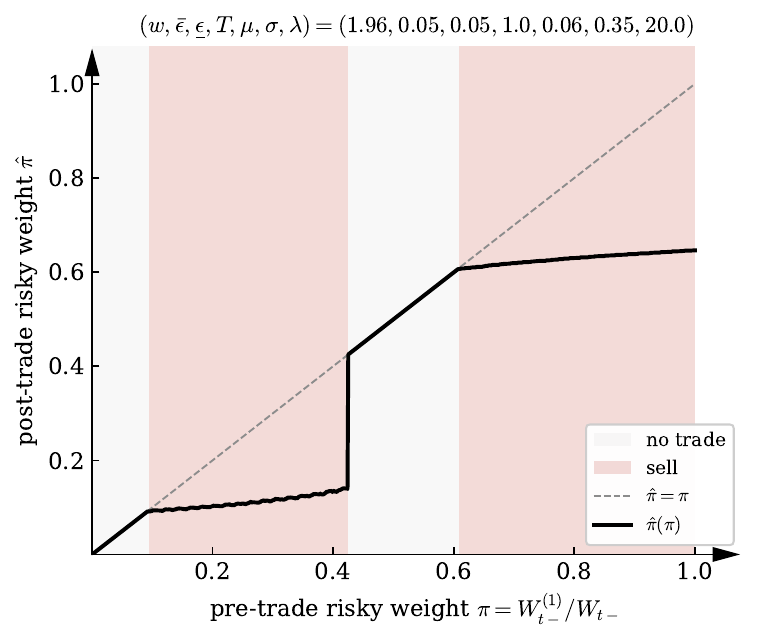}
\caption{Numerically computed optimal post-trade risky weight $\hat{\pi}$ as a function of the pre-trade risky weight $\pi$, evaluated at $t=0$ for $U=U_1$, $\overline{\epsilon}=\underline{\epsilon}=0.05$, and pre-trade wealth $w=1.51$ (left) and $w=1.96$ (right). For $w=1.51$, the selected actions follow the usual buy--no-trade--sell ordering. For $w=1.96$, two no-trade intervals alternate with two selling intervals, and no buying is selected.}
\label{Utility_function_1}
\end{figure}

\begin{figure}[ht]
\centering
\includegraphics[width=.45\columnwidth]{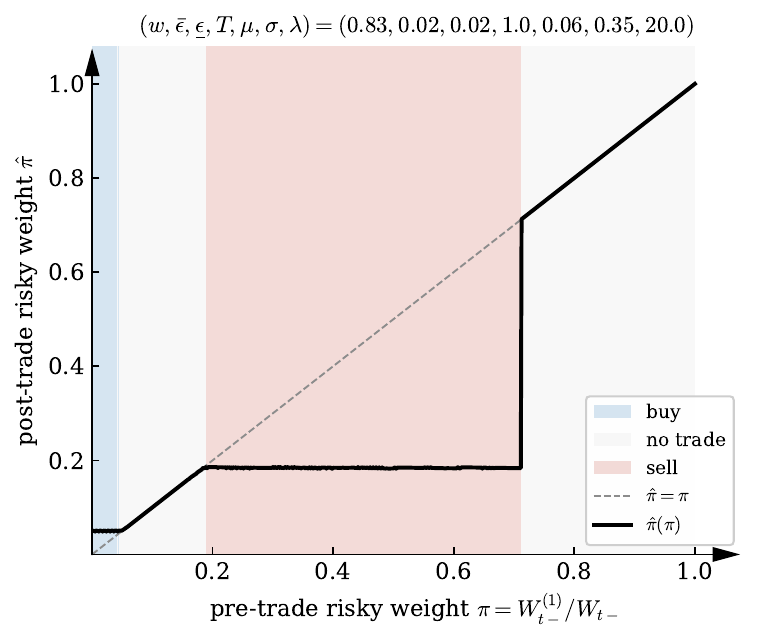} \quad \includegraphics[width=.45\columnwidth]{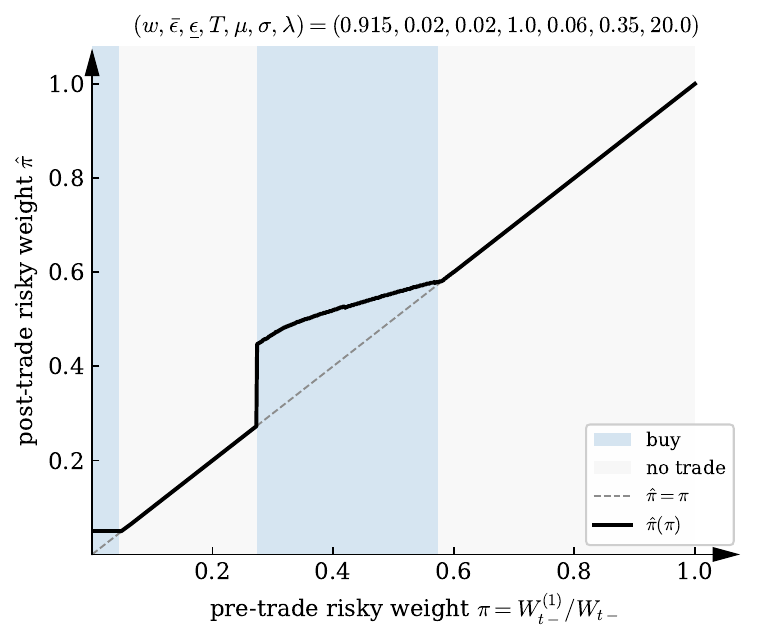}
\caption{Numerically computed optimal post-trade risky weight $\hat{\pi}$ as a function of the pre-trade risky weight $\pi$, evaluated at $t=0$ for $U=U_2$, $\overline{\epsilon}=\underline{\epsilon}=0.02$, and pre-trade wealth $w=0.83$ (left) and $w=0.915$ (right). For $w=0.83$, the usual buy--no-trade--sell pattern is followed by an additional no-trade interval at high pre-trade risky weights. For $w=0.915$, two buying intervals alternate with two no-trade intervals, and no selling is selected.}
\label{Utility_function_2}
\end{figure}

\begin{figure}[ht]
\centering
\includegraphics[width=.45\columnwidth]{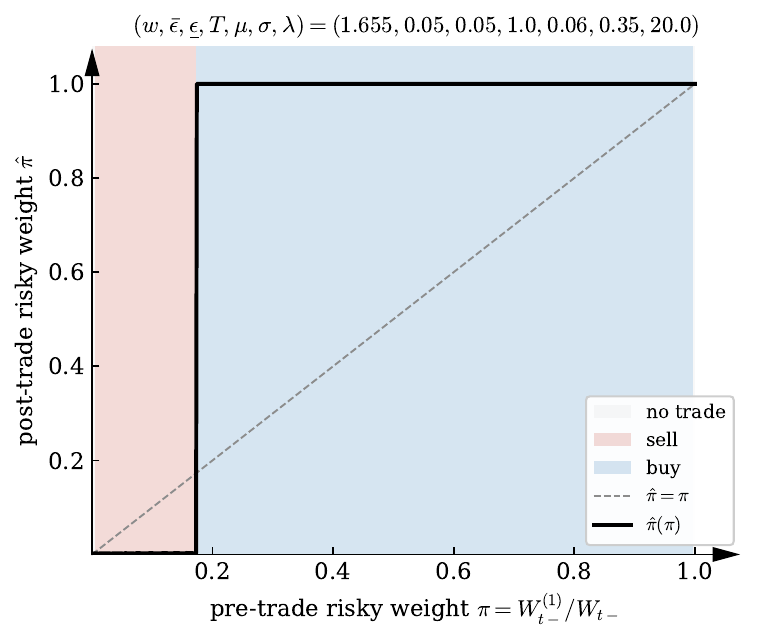}
\caption{Numerically computed optimal post-trade risky weight $\hat{\pi}$ as a function of the pre-trade risky weight $\pi$, evaluated at $t=0$ for $U=U_3$, $\overline{\epsilon}=\underline{\epsilon}=0.05$, and pre-trade wealth $w=1.655$. As $\pi$ increases through $(0,1)$, the selected action switches from selling to buying, reversing the usual ordering. At the endpoints $\pi=0$ and $\pi=1$, the displayed policy satisfies $\hat{\pi}=\pi$ and therefore selects no trade.}
\label{Utility_function_3}
\end{figure}

We comment on the economic interpretation of $U_1,U_2,U_3$. Each of them is the sum of a CRRA-type component ($-w^{-1}$ or $-w^{-2}$, i.e.,\ relative risk aversion $2$ or $3$) and one or two smoothed upward steps of the form $h\,\Phi((w-\bar w)/\eta)$ with height $h>0$, level $\bar w>0$, and width $\eta>0$. For every $w\neq \bar w$, $\lim_{\eta\downarrow 0} \Phi((w-\bar w)/\eta)=\mathbf 1_{\{w>\bar w\}}$. For instance, set
\[
U_{2,\eta}(w):=-\frac{1}{w^{2}}+8\,\Phi\Big(\frac{w-1.8}{\eta}\Big),\qquad \eta>0,
\]
so that $U_2=U_{2,0.15}$. As $\eta\downarrow 0$, $U_{2,\eta}$ converges pointwise on $(0,\infty)\setminus\{1.8\}$ to $-w^{-2}+8\cdot\mathbf 1_{\{w>1.8\}}$, a CRRA utility plus an upward jump at the aspiration level $1.8$. Such discontinuous limits are aspiration utilities in the sense of \cite{DiecidueVandeVen2008}, which have been applied to investment decisions in \cite{AristidouGigaLeeZapatero2025} and to portfolio choice in \cite{DaiKouQianWan2022,QianYang2026}; $U_2$ may thus be viewed as a smoothed aspiration utility. Without the CRRA component, the limiting step $\mathbf 1_{\{w>\bar w\}}$ corresponds to maximizing the probability of reaching the target wealth $\bar w$, as in \cite{Browne1999}.
 
The convexity intervals lie close to, but need not end exactly at, the aspiration levels because the curvature of each utility also includes the CRRA component and the contributions of the other smoothed steps. Numerically, the convexity intervals are approximately $(0.87,1.2)$ and $(1.48,1.8)$ for $U_1$, $(1.31,1.8)$ for $U_2$, and $(1.51,1.6)$ and $(1.87,2.1)$ for $U_3$; the utilities are concave outside these intervals. Hence $U_2$ is concave--convex--concave, the shape proposed by \cite{FriedmanSavage1948}, while $U_1$ and $U_3$, which have two aspiration levels, are concave--convex--concave--convex--concave. This recalls the alternating curvature of the utility of \cite{Markowitz1952}, which, however, has three rather than four inflection points and is defined over gains and losses relative to customary wealth rather than over terminal wealth. The coexistence of a low and a high aspiration level is also reminiscent of the layered portfolios of the multiple account version of the behavioral portfolio theory of \cite{ShefrinStatman2000}, which builds on the security-potential/aspiration theory of \cite{Lopes1987}; there, however, the two aspiration levels belong to separate mental accounts rather than to a single utility function. Several wealth goals also appear in the goal-based wealth management framework of \cite{CapponiZhang2024}, where they are pursued sequentially over time.
 
The smoothing is needed for our analysis: a genuine jump in $U$ violates the H\"older condition in Assumption~\ref{U_assumption}, and discontinuous utilities of this type have been treated with discontinuous viscosity solutions in \cite{DaiKouQianWan2022,QianYang2026}. For each fixed $\eta>0$, the smoothed utilities stay within the classical solution theory of Section~4, while small widths, such as $\eta=0.02$ in the first step of $U_3$, make them close to the corresponding step utilities. Since the constants in the proof of Lemma~\ref{lem:value_regularity} depend, through the bound on the terminal data $g$ in (iv) of Lemma~\ref{lem:coeff}, on the H\"older norm in Assumption~\ref{U_assumption}, and this norm blows up as $\eta\downarrow 0$, we do not claim any convergence of the value functions or policies in this limit.
 
The buy--no-trade--sell structure, established for CRRA utilities under continuous trading by \cite{DavisNorman1990,MagillConstantinides1976,ShreveSoner1994} and under search frictions by \cite{TJ2026}, can break down for non-concave utilities: \cite{QianYang2026} finds disjoint no-trade regions, and without transaction costs \cite{DaiKouQianWan2022} shows examples of discontinuous optimal strategies. In our model, with no borrowing or short-selling, Figures~2--4 show analogues along fixed-wealth slices: several disjoint intervals on which the same action is selected (right of Figure~2, Figure~3); departures from the usual ordering, most strikingly a selling interval to the left of a buying interval (Figure~4); and jumps in the selected post-trade weight (right of Figure~2, Figures~3--4), including in Figure~4 a switch between the bounds $\hat\pi=0$ and $\hat\pi=1$, a counterpart under transaction costs and search frictions of the discontinuities in \cite{DaiKouQianWan2022}.
In Figure~4, the pre-trade wealth $w=1.655$ lies just above the first aspiration level $1.6$ of $U_3$ and below the second level $2.1$. The selling interval can thus be read as securing the first level, and the buying interval with $\hat\pi=1$ as gambling for the second. Consistently, the switch near $\pi\approx 0.18$ occurs roughly where a full sale would leave $w(1-\underline{\epsilon}\,\pi)\approx 1.64$, two step widths above the first level.
 
Such jumps are made possible by the non-concavity of $U$. At a trading time, the new position maximizes $z\mapsto \overline K(\widetilde\nu)(t,m,x,z)$ over the whole segment of reachable post-trade positions (Theorem~\ref{Verification} and Corollary~\ref{Optimal_portfolio}); when $U$ is not concave, this map may have several local maximizers, and as $\pi$ varies the global maximizer can move from one to another. We call $\pi$, for fixed $t$ and $w$, a switching point if $\operatorname*{argmax}_{z\in[0,1]}\overline K(\widetilde\nu)(t,w\pi,w(1-\pi),z)$ contains at least two maximizers with different post-trade weights. Under continuous trading with concave $U$, the state leaves the closure of the no-trade region only at the initial time; here it may lie deep inside a trading region at every arrival, so such switches can recur.
 
Moreover, every jump of $\hat\pi$ occurs at a switching point. Fix $t$ and $w$, let $\hat\pi(\pi)$ be the post-trade weight of the selected maximizer, and suppose that $\hat\pi$ is discontinuous at $\pi_0$. By compactness, there exist $\pi_n\to\pi_0$ with $\hat\pi(\pi_n)\to\bar\pi\neq\hat\pi(\pi_0)$ and, along a subsequence, corresponding maximizers $\hat z_n\to\bar z\in[0,1]$. Since $\widetilde\nu(t,m,x)=\nu(t,l,x,me^{-l})$ for any fixed $l\in\mathbb R$, Lemma~\ref{lem:value_regularity} makes $\widetilde\nu(t,\cdot,\cdot)$ continuous on $\mathbb D$, including at $m=0$, so $(\pi,z)\mapsto \overline K(\widetilde\nu)(t,w\pi,w(1-\pi),z)$ is continuous on $[0,1]^2$ and $\bar z$ is a maximizer at $\pi_0$, as in the proof of Lemma~\ref{lem:measurable}. As $m_z/(x_z+m_z)$ is continuous in $(\pi,z)$, because $x_z+m_z\ge(1-\epsilon_M)w>0$ by (i) of Lemma~\ref{lem:post_trade}, $\bar z$ yields $\bar\pi\neq\hat\pi(\pi_0)$, so $\pi_0$ is a switching point. The abrupt changes in Figures~\ref{Utility_function_1}--\ref{Utility_function_3} are consistent with switches between competing global maximizers. Distinguishing such switches from discretization effects requires numerical refinement and comparison of the competing objective values.

\bigskip

\appendix
\section*{Appendix}
\refstepcounter{section}

\begin{lemma}[Borel measurable selection theorem] \label{lem:measurable}
Let $F:[0,T] \times \R \times \D \times [0,1] \to \R$ be a continuous function. Then the function 
\begin{align}
(t,l,x,y) \mapsto f(t,l,x,y):=\max \bigg\{ z: z \in \argmax_{z \in [0,1]} F(t,l,x,y,z) \bigg\},
\end{align}  
is upper semicontinuous (hence Borel measurable). 
\end{lemma}

\begin{proof}
Note first that $f$ is well defined: for each $(t,l,x,y)$ the set $\argmax_{z \in [0,1]} F(t,l,x,y,z)$ is nonempty and compact by continuity of $F$ and compactness of $[0,1]$, so it contains its supremum. Let $(t_{n},l_{n},x_{n},y_{n}) \to (t,l,x,y)$ in $[0,T] \times \R \times \D$ and set $z_{n} := f(t_{n},l_{n},x_{n},y_{n})$. Passing to a subsequence, we may assume $z_{n} \to z^{*} \in [0,1]$ and $z_{n} \to \limsup_{n} f(t_{n},l_{n},x_{n},y_{n})$. Since the points $(t_{n},l_{n},x_{n},y_{n},z_{n})$ eventually lie in a compact subset of $[0,T] \times \R \times \D \times [0,1]$, on which $F$ is uniformly continuous, the maximum function $(t,l,x,y) \mapsto \max_{z \in [0,1]} F(t,l,x,y,z)$ is continuous there, and 
\[
F(t,l,x,y,z^{*}) = \lim_{n \to \infty} F(t_{n},l_{n},x_{n},y_{n},z_{n}) = \lim_{n \to \infty} \max_{z \in [0,1]} F(t_{n},l_{n},x_{n},y_{n},z) = \max_{z \in [0,1]} F(t,l,x,y,z).
\]
Hence $z^{*} \in \argmax_{z \in [0,1]} F(t,l,x,y,z)$ and therefore $\limsup_{n} f(t_{n},l_{n},x_{n},y_{n}) = z^{*} \le f(t,l,x,y)$, which is upper semicontinuity. See also Lemma F.2 of \cite{TJ2026} for a closely related statement.
\end{proof}

\begin{notation}\label{not:master}
From now on, we introduce the notation used in the lemmas below. For $\bx=(l,x,y)\in \R \times \D$, write
\[
w(\bx):=x+ye^{l},\qquad
\pi(\bx):=\frac{ye^{l}}{w(\bx)}\in[0,1],\qquad
r(\bx):=\ln w(\bx),\qquad
F(r):=\ln\cU(e^{r}),\ r \in \R,
\]
and denote by
\[
\rho_{0}(\bx_{1},\bx_{2})
:=|l_{1}-l_{2}|
+\frac{|x_{1}-x_{2}|+e^{l_{1}\vee l_{2}}|y_{1}-y_{2}|}
       {w(\bx_{1})\wedge w(\bx_{2})}
\]
the spatial part of $\rho$, so that
$\rho := \rho((t_{1},\bx_{1}),(t_{2},\bx_{2}))=\sqrt{|t_{1}-t_{2}|}+\rho_{0}(\bx_{1},\bx_{2})$.
For $\bp_{i}=(x_{i},y_{i})$ ($i \in \{ 1,2 \}$) and $l \in \R$, we set
\[
D_{\bp_{1},\bp_{2}}(l)
:=\frac{|x_{1}-x_{2}|+e^{l}|y_{1}-y_{2}|}
        {(x_{1}+y_{1}e^{l})\wedge(x_{2}+y_{2}e^{l})}.
\]
We write $K_{\be}$ and $K_{\ue}$ for the Lipschitz constants of $l\mapsto\be(e^{l})$ and $l\mapsto\ue(e^{l})$, which are finite by Assumption \ref{Assumption_1}.
From $\cU(w)=1+w^{p}+w^{-p}$, we have the elementary bounds
\begin{equation}\label{eq:Fbounds}
  \bigg\| \frac{d^{k} F}{d r^{k}} \bigg\|_{\infty} < \infty \quad \text{ for all } k \in \{ 1,2,3\}.
\end{equation}
Also we have the identity
\begin{equation}\label{eq:master}
  e^{l} \theta(l,x,y,z) = w(\bx)\,q(l,\pi,z), \quad \text{ where } \quad q(l,\pi,z):=\frac{(1+\be \pi)(z-b)}{1+\be}
\in\Big[-\pi,\ \tfrac{1-\pi}{1+\be}\Big],\quad z\in[0,1],
\end{equation}
with simplified notations $\ue=\ue(e^{l})$, $\be=\be(e^{l})$, and
\[
b=b(\bx)=\frac{ye^{l}[1+\be(e^{l})]}{ye^{l}[1+\be(e^{l})]+x}=\frac{(1+\be)\pi}{1+\be \pi},
\]
where the two expressions for $b$ coincide upon dividing the numerator and the denominator of the latter by $w(\bx)$.

For simplicity, 
\begin{align*}
  L^i & := L^{(t_i,\mathbf{x}_i)}, \quad \rho := \rho((t_{1}, \mathbf{x}_{1}), (t_{2}, \mathbf{x}_{2})), \\[2pt]
  A_i(u) &:= \exp \bigg( \int_{t_i}^{u}\!\Big[-\alpha-\lambda\big(a,e^{L^{i}_a}\big)
        +\gamma\big(a,L^{i}_a,x_i,y_i\big)\Big]\,da \bigg)
        \lambda\big(u,e^{L^{i}_u}\big),
         \\[2pt]
  K^{(n)}_i(u) &:= \sup_{z\in[0,1]} K(w_n)\big( u,L^{i}_u,x_i,y_i,z \big), \\[2pt]
  \mathcal{E}_{i}(u) &:= \exp \bigg( \int_{t_{i}}^{u}\big[-\alpha-\lambda(a,e^{L^{i}_{a}})
+\gamma(a,L^{i}_{a},x_{i},y_{i})\big]\,da \bigg).
\end{align*}
\end{notation}

The next lemma isolates, once and for all, the mean-value argument along segments in $\D$ that we repeatedly need for functions of the intrinsic variables $(l,\pi,r)$.

\begin{lemma}[Composition estimates]\label{lem:composition}
Let $c_{A}>0$ and let $A:\R\times[0,1]\times\R\times[0,1]\to\R$, $(l,\pi,r,z)\mapsto A(l,\pi,r,z)$, satisfy $\|A\|_{\infty}\le c_{A}$ and be Lipschitz continuous in each of its arguments with constant $c_{A}$. For $\bx=(l,x,y)\in\R\times\D$ and $z\in[0,1]$, set
\[
\Psi(\bx,z):=A\big(l,\pi(\bx),r(\bx),z\big).
\]
Then for all $\bx_{1},\bx_{2}\in\R\times\D$ with $\bp_{i}=(x_{i},y_{i})$ and all $z,z'\in[0,1]$:
\begin{enumerate}
    \item[(i)] for every $l\in\R$ with $(l,\bp_{1}),(l,\bp_{2})\in\R\times\D$,
    \begin{equation}\label{eq:comp-fixed-l}
    \big|\Psi((l,\bp_{1}),z)-\Psi((l,\bp_{2}),z)\big|
    \ \le\ 2c_{A}\,\min\big\{1,\,D_{\bp_{1},\bp_{2}}(l)\big\};
    \end{equation}
    \item[(ii)] for every $l\in\R$, $j \in \{ 1,2 \}$, and $h'\in\R$,
    \begin{equation}\label{eq:comp-shift}
    D_{\bp_{1},\bp_{2}}(l+h')\ \le\ e^{2|h'|}\,D_{\bp_{1},\bp_{2}}(l),
    \qquad
    D_{\bp_{1},\bp_{2}}(l_{j})\ \le\ e^{|l_{1}-l_{2}|}\,\rho_{0}(\bx_{1},\bx_{2}),
    \end{equation}
    and consequently, for every $l\in\R$,
    \begin{equation}\label{eq:comp-min}
    \big|\Psi((l,\bp_{1}),z)-\Psi((l,\bp_{2}),z)\big|
    \ \le\ 2e\,c_{A}\,\min\big\{1,\,e^{2|l-l_{2}|}\rho_{0}(\bx_{1},\bx_{2})\big\};
    \end{equation}
    \item[(iii)] for every $\bp=(x,y)$ with $(l',\bp),(l'',\bp)\in\R\times\D$,
    \[
    \big|\Psi((l',\bp),z)-\Psi((l'',\bp),z)\big|\ \le\ 3c_{A}\,|l'-l''|;
    \]
    \item[(iv)] if $\rho_{0}(\bx_{1},\bx_{2})\le 1$, then
    \begin{equation*}
    \big|\Psi(\bx_{1},z)-\Psi(\bx_{2},z)\big|\ \le\ (2e+3)\,c_{A}\,\rho_{0}(\bx_{1},\bx_{2}).
    \end{equation*}
    Furthermore, suppose that $A$ is bounded by $c_{A}$ and, instead of being Lipschitz, $\delta$-H\"older continuous in each of its arguments with constant $c_{A}$, for some $\delta \in (0,1]$. Then $\rho_{0}(\bx_{1},\bx_{2})\le 1$ implies 
    \begin{equation*}
    \big|\Psi(\bx_{1},z)-\Psi(\bx_{2},z)\big|\ \le\ 3^{1 - \delta} (2e+3)^{\delta}\,c_{A}\,[\rho_{0}(\bx_{1},\bx_{2})]^{\delta};
    \end{equation*}
    \item[(v)] $\big|\Psi(\bx,z)-\Psi(\bx,z')\big|\le c_{A}|z-z'|$.
\end{enumerate}
\end{lemma}

\begin{proof}
(i) The bound by $2c_{A}\ge 2\|A\|_{\infty}$ is trivial, so it suffices to bound the increment by $2c_{A}D_{\bp_{1},\bp_{2}}(l)$. At fixed $l$ we have
\begin{equation} \label{pi_r_derivatives_1}
  \bigg| \frac{\partial \pi}{\partial x} \bigg| = \frac{\pi}{w} \le \frac{1}{w}, \quad
\bigg| \frac{\partial \pi}{\partial y} \bigg| = \frac{e^{l}(1-\pi)}{w} \le \frac{e^{l}}{w}, \quad
\bigg| \frac{\partial r}{\partial x} \bigg| = \frac{1}{w}, \quad
\bigg| \frac{\partial r}{\partial y} \bigg| = \frac{e^{l}}{w}.
\end{equation}
The segment $\bp_{\tau}:=(1-\tau)\bp_{2}+\tau\bp_{1}$, $\tau\in[0,1]$, has nonnegative coordinates and, since $w(l,\cdot)$ is affine, satisfies $w(l,\bp_{\tau})\ge(x_{1}+y_{1}e^{l})\wedge(x_{2}+y_{2}e^{l})>0$; in particular $\bp_{\tau}\in\D$. Hence the mean value theorem along $\bp_{\tau}$ together with the Lipschitz continuity of $A$ in $(\pi,r)$ yields
\[
  \big|\Psi((l,\bp_{1}),z)-\Psi((l,\bp_{2}),z)\big|
\le c_{A}\big(|\pi(l,\bp_{1})-\pi(l,\bp_{2})|+|r(l,\bp_{1})-r(l,\bp_{2})|\big)
\le 2c_{A}\,D_{\bp_{1},\bp_{2}}(l).
\]

(ii) For any $h'\in\R$, the numerator of $D_{\bp_{1},\bp_{2}}(l+h')$ is at most $e^{|h'|}$ times that of $D_{\bp_{1},\bp_{2}}(l)$, and each term $x_{i}+y_{i}e^{l+h'}$ of the denominator is at least $e^{-|h'|}(x_{i}+y_{i}e^{l})$; this proves the first display of \eqref{eq:comp-shift}. For the second, the numerator of $D_{\bp_{1},\bp_{2}}(l_{j})$ is at most $|x_{1}-x_{2}|+e^{l_{1}\vee l_{2}}|y_{1}-y_{2}|$, while $x_{i}+y_{i}e^{l_{j}}\ge e^{-|l_{1}-l_{2}|}(x_{i}+y_{i}e^{l_{i}})\ge e^{-|l_{1}-l_{2}|}\big(w(\bx_{1})\wedge w(\bx_{2})\big)$. For \eqref{eq:comp-min}: if $e^{2|l-l_{2}|}\rho_{0}\ge1$, the minimum equals $1$ and \eqref{eq:comp-min} follows from \eqref{eq:comp-fixed-l}. Otherwise $\rho_{0}<1$, hence $|l_{1}-l_{2}|\le\rho_{0}<1$ (as $|l_{1}-l_{2}|$ is a summand of $\rho_{0}$), and \eqref{eq:comp-shift} gives
$D_{\bp_{1},\bp_{2}}(l)\le e^{2|l-l_{2}|}D_{\bp_{1},\bp_{2}}(l_{2})\le e^{2|l-l_{2}|}e^{|l_{1}-l_{2}|}\rho_{0}\le e\,e^{2|l-l_{2}|}\rho_{0}$,
so \eqref{eq:comp-min} follows from \eqref{eq:comp-fixed-l}.

(iii) At fixed $\bp$, the maps $l\mapsto\pi(l,\bp)$ and $l\mapsto r(l,\bp)$ satisfy 
\begin{equation}  \label{pi_r_derivatives_2}
    \Big|\frac{\partial \pi}{\partial l}\Big|=\pi(1-\pi)\le 1, \quad \Big|\frac{\partial r}{\partial l}\Big|=\pi \le 1,
\end{equation}
so, accounting also for the explicit $l$-dependence of $A$,
$\big|\Psi((l',\bp),z)-\Psi((l'',\bp),z)\big|\le 3 c_{A} |l'-l''|$.

(iv) We recall the notation $\bp_{i}=(x_{i},y_{i})$ for $i \in \{ 1,2 \}$ again. We use \eqref{pi_r_derivatives_1} for $\pi(l_{1}, \bp_{1}) - \pi(l_{1}, \bp_{2})$ and $r(l_{1}, \bp_{1}) - r(l_{1}, \bp_{2})$, and use \eqref{pi_r_derivatives_2} for $\pi(l_{1}, \bp_{2}) - \pi(l_{2}, \bp_{2})$ and $r(l_{1}, \bp_{2}) - r(l_{2}, \bp_{2})$. In addition to using $|l_{1}-l_{2}|\le\rho_{0}\le1$ so that $D_{\bp_{1},\bp_{2}}(l_{1})\le e\,\rho_{0}$, we have 
\[
  |l_{1} - l_{2}| + |\pi(\bx_{1}) - \pi(\bx_{2})| + |r(\bx_{1}) - r(\bx_{2})| \le (2e+3) \rho_{0}(\bx_{1}, \bx_{2}).
\]
Therefore, 
\[
  |\Delta \Psi| \le c_{A} (|\Delta l|^{\delta} + |\Delta \pi|^{\delta} + |\Delta r|^{\delta}) \le 3^{1-\delta}(2e+3)^{\delta} c_{A} \rho_{0}^{\delta}.
\]
The case $\delta = 1$ takes charge of the Lipschitz hypothesis.

(v) This proof is immediate from the Lipschitz continuity of $A$ in $z$.
\end{proof}

\begin{lemma}\label{lem:post_trade}
Let 
\begin{align}
    J_{1}(l,x,y,z) &:= x - [1+\overline\epsilon(e^{l})] e^{l} \theta(l,x,y,z) = x - [1+\overline\epsilon(e^{l})] \bigg[ z \bigg( \frac{x}{1 + \overline{\epsilon}(e^{l})} + y e^{l} \bigg) - y e^{l} \bigg], \label{J1} \\
    J_{2}(l,x,y,z) &:= x - [1-\underline\epsilon(e^{l})] e^{l} \theta(l,x,y,z) = x - [1-\underline\epsilon(e^{l})] \bigg[ z \bigg( \frac{x}{1 + \overline{\epsilon}(e^{l})} + y e^{l} \bigg) - y e^{l} \bigg], \label{J2} \\
    J_{3}(l,x,y,z) &:= y+\theta(l,x,y,z) = z \bigg( \frac{x e^{-l}}{1 + \overline{\epsilon}(e^{l})} + y \bigg), \label{J3} \\
    U_{1}(l,x,y,z) &:= \frac{\cU \big( x + y e^{l} - \overline{\epsilon}(e^{l}) e^{l} \theta(l,x,y,z) \big)}{\cU \big( x + y e^{l} \big)} = \frac{\cU \Big( x + y e^{l} - \overline{\epsilon}(e^{l}) \Big[ z \Big( \frac{x}{1 + \overline{\epsilon}(e^{l})} + y e^{l} \Big) - y e^{l} \Big] \Big)}{\cU \big( x + y e^{l} \big)}, \\
    U_{2}(l,x,y,z) &:= \frac{\cU \big( x + y e^{l} + \underline{\epsilon}(e^{l}) e^{l} \theta(l,x,y,z) \big)}{\cU \big( x + y e^{l} \big)} = \frac{\cU \Big( x + y e^{l} + \underline{\epsilon}(e^{l}) \Big[ z \Big( \frac{x}{1 + \overline{\epsilon}(e^{l})} + y e^{l} \Big) - y e^{l} \Big] \Big)}{\cU \big( x + y e^{l} \big)}.
\end{align}
For any $\bx \in \mathbb{R} \times \D$ and $z\in[0,1]$,
\begin{enumerate}
\item[(i)] the post-trade to pre-trade wealth ratios satisfy
\begin{align*}
\frac{x+ye^{l}-\be(e^{l})e^{l}\theta(l,x,y,z)}{w(\bx)} &= 1-\be q(l,\pi,z)
 \ \in\ \Big[1-\tfrac{\be(1-\pi)}{1+\be},\ 1+\be\pi\Big] \subset \big[(1+\eM)^{-1},\,1+\eM\big], \\
\frac{x+ye^{l}+\ue(e^{l})e^{l}\theta(l,x,y,z)}{w(\bx)} &= 1+\ue q(l,\pi,z)
 \ \in\ \Big[1-\ue\pi,\ 1+\tfrac{\ue(1-\pi)}{1+\be}\Big] \subset \big[1-\eM,\,1+\eM\big];
\end{align*}
in particular, 
\[
    (1-\epsilon_M)w(\mathbf{x}) \le J_i(\mathbf{x},z)+e^lJ_3(\mathbf{x},z) \le (1+\epsilon_M)w(\mathbf{x}), \qquad i\in\{1,2\};
\]

\item[(ii)] $J_{1}(\bx,z), J_{2}(\bx,z), J_{3}(\bx,z) \ge 0$ and
$(J_{i}(\bx,z),J_{3}(\bx,z))\in \D$ for $i \in \{ 1,2 \}$;
\item[(iii)] $\|U_{1}\|_{\infty}\vee\|U_{2}\|_{\infty}\le c_{\cU} :=(1+\eM)^{p}+(1-\eM)^{-p}$;
\item[(iv)] if $\rho_{0}(\bx_{1},\bx_{2})\le 1$, then for $i\in\{1,2\}$,
\begin{equation}
\rho_{0}\big((l_{1},J_{i}(\bx_{1},z),J_{3}(\bx_{1},z)),\,
             (l_{2},J_{i}(\bx_{2},z),J_{3}(\bx_{2},z))\big)
\ \le\ \hat C'\,\rho_{0}(\bx_{1},\bx_{2}); \label{rho_inequality}
\end{equation}
\item[(v)]
$\rho_{0}\big((l,J_{1}(\bx,z),J_{3}(\bx,z)),\,
              (l,J_{2}(\bx,z),J_{3}(\bx,z))\big) \le \hat C'\,|z-b(\bx)|$;
\item[(vi)] If $\rho_0(\mathbf{x}_1,\mathbf{x}_2)\le 1$, then
\begin{equation*}
    \bigl|b(\mathbf{x}_1)-b(\mathbf{x}_2)\bigr| \le \hat{C}'\,\rho_0(\mathbf{x}_1,\mathbf{x}_2);
\end{equation*}
\item[(vii)] If $\rho_0(\mathbf{x}_1,\mathbf{x}_2) \le 1$, then for $i\in\{1,2\}$ and every $z\in[0,1]$,
\begin{equation*}
    \bigl|U_i(\mathbf{x}_1,z)-U_i(\mathbf{x}_2,z)\bigr| \le \hat{C}'\,\rho_0(\mathbf{x}_1,\mathbf{x}_2);
\end{equation*}
\item[(viii)] for $i\in\{1,2\}$, every $\mathbf{x}\in\mathbb{R}\times \D$ and $z,z'\in[0,1]$,
  \begin{equation*}
    \bigl|U_i(\mathbf{x},z)-U_i(\mathbf{x},z')\bigr| \le \hat{C}'\,|z-z'|
  \end{equation*}
\end{enumerate}
for constant $\hat C'>0$.
\end{lemma}

\begin{proof}
(i) By \eqref{eq:master} and direct computation,
\begin{equation} \label{eq:posttrade-wealth}
  \begin{split}
  & J_{1}(\bx,z)+e^{l}J_{3}(\bx,z)=x+ye^{l}-\be e^{l}\theta=w(\bx)\big(1-\be q\big), \\
  & J_{2}(\bx,z)+e^{l}J_{3}(\bx,z)=x+ye^{l}+\ue e^{l}\theta=w(\bx)\big(1+\ue q\big),
  \end{split}
\end{equation}
with $q=q(l,\pi,z)\in\big[-\pi,\frac{1-\pi}{1+\be}\big]$. The stated ranges of $1-\be q$ and $1+\ue q$ follow immediately, and the inclusions follow from
$1-\frac{\be(1-\pi)}{1+\be}\ge1-\frac{\be}{1+\be}=(1+\be)^{-1}\ge(1+\eM)^{-1}$
together with $\be\pi\le\eM$, $\ue\pi\le\eM$, and $\frac{\ue(1-\pi)}{1+\be}\le\eM$.

(ii) From \eqref{J1}, \eqref{J2}, \eqref{J3},
$J_{1}=(1-z)\big(x+(1+\be)ye^{l}\big) \ge 0$,
$J_{3}=z\big(\frac{xe^{-l}}{1+\be}+y\big) \ge 0$, and
$J_{2}=x\big(1-\frac{(1-\ue)z}{1+\be}\big)+(1-\ue)(1-z)ye^{l} \ge 0$
since $\frac{(1-\ue)z}{1+\be} \le 1$. Moreover, $(J_{i}(\bx,z),J_{3}(\bx,z))\in \D$ follows from the identities \eqref{eq:posttrade-wealth} above together with (i), which show that the post-trade total wealth $J_{i}(\bx,z)+e^{l}J_{3}(\bx,z)$ is strictly positive.

(iii) By (i) and the definitions of $U_{1}, U_{2}$, for each $(\bx,z)$ there exists $c\in[1-\eM,1+\eM]$ such that $U_{i}(\bx,z)=\cU(c\,w(\bx))/\cU(w(\bx))$. Since 
\[
    \cU(cw)=1+c^{p}w^{p}+c^{-p}w^{-p}
    \le (1+\eM)^{p}(1+w^{p})+(1-\eM)^{-p}w^{-p}
    \le c_{\cU}\,\cU(w),
\]
the claim follows.

(iv) The numerator of the left-hand side of \eqref{rho_inequality} is
$|J_{i}(\bx_{1},z)-J_{i}(\bx_{2},z)|+e^{l_{1}\vee l_{2}}|J_{3}(\bx_{1},z)-J_{3}(\bx_{2},z)|$.
We first claim that, whenever $|l_{1}-l_{2}|\le1$, this numerator is bounded by
\begin{equation}\label{eq:J-increment}
\hat C\big(|x_{1}-x_{2}|+e^{l_{1}\vee l_{2}}|y_{1}-y_{2}|
+(x_{2}+y_{2}e^{l_{2}})|l_{1}-l_{2}|\big)
\end{equation}
for a constant $\hat C=\hat C(K_{\be},K_{\ue})$. Indeed, by the proof of (ii),
\[
J_{1}=(1-z)\big(x+\varphi(l)y\big),\qquad
J_{2}=\big(1-z\kappa(l)\big)x+(1-z)\chi(l)y,\qquad
J_{3}=z\big(\psi(l)x+y\big),
\]
where $\varphi(l):=[1+\be(e^{l})]e^{l}$, $\chi(l):=[1-\ue(e^{l})]e^{l}$, $\psi(l):=e^{-l}/[1+\be(e^{l})]$, and $\kappa(l):=[1-\ue(e^{l})]/[1+\be(e^{l})]$. By Assumption \ref{Assumption_1},
\[
|\varphi'|\le(2+K_{\be})e^{l},\qquad
|\chi'|\le(2+K_{\ue})e^{l},\qquad
|\psi'|\le(1+K_{\be})e^{-l},\qquad
|\kappa'|\le K_{\ue}+K_{\be}.
\]
Since $|l_{1}-l_{2}|\le1$, every intermediate point $l$ between $l_{1}$ and $l_{2}$ satisfies $e^{l}\le e\,e^{l_{2}}$ and $e^{-l}\le e\,e^{-(l_{1}\wedge l_{2})}$, so the mean value theorem in $l$ gives, for instance,
\begin{align*}
  |J_{1}(\bx_{1},z)-J_{1}(\bx_{2},z)| & \le |x_{1}-x_{2}| + \varphi(l_{1}) |y_{1}-y_{2}|+y_{2}|\varphi(l_{1})-\varphi(l_{2})| \\
  & \le |x_{1}-x_{2}|+2e^{l_{1}\vee l_{2}}|y_{1}-y_{2}|+(2+K_{\be})e\,y_{2}e^{l_{2}}|l_{1}-l_{2}|,
\end{align*}
and using $e^{l_{1}\vee l_{2}}\psi(l_{1})\le e^{|l_{1}-l_{2}|}\le e$ and $e^{l_{1}\vee l_{2}}e^{-(l_{1}\wedge l_{2})}=e^{|l_{1}-l_{2}|}\le e$,
\[
e^{l_{1}\vee l_{2}}|J_{3}(\bx_{1},z)-J_{3}(\bx_{2},z)|
\le e\,|x_{1}-x_{2}|+(1+K_{\be})e^{2}\,x_{2}|l_{1}-l_{2}|+e^{l_{1}\vee l_{2}}|y_{1}-y_{2}|.
\]
The estimate for $J_{2}$ is analogous, which proves the claim.

Next, assume without loss of generality $w(\bx_{1})=w(\bx_{1})\wedge w(\bx_{2})$. When $\rho_{0}:=\rho_{0}(\bx_{1},\bx_{2})\le1$ we have $|l_{1}-l_{2}|\le1$ and hence 
\begin{align*}
  w(\bx_{2}) &= x_{2}+y_{2}e^{l_{2}}
\le x_{1}+y_{1}e^{l_{2}}+|x_{1}-x_{2}|+e^{l_{1}\vee l_{2}}|y_{1}-y_{2}| \\
  & \le e^{|l_{1}-l_{2}|}w(\bx_{1})+\rho_{0}\big(w(\bx_{1})\wedge w(\bx_{2})\big)
\le(e+1)\big(w(\bx_{1})\wedge w(\bx_{2})\big),
\end{align*}
where the second inequality uses $e^{l_{2}}\le e^{|l_{1}-l_{2}|}e^{l_{1}}$ and the definition of $\rho_{0}$, and the third uses $\rho_{0}\le1$. Combining this with \eqref{eq:J-increment} and $|x_{1}-x_{2}|+e^{l_{1}\vee l_{2}}|y_{1}-y_{2}|\le\rho_{0}\big(w(\bx_{1})\wedge w(\bx_{2})\big)$, the numerator of the left-hand side of \eqref{rho_inequality} is bounded by
$C\big(w(\bx_{1})\wedge w(\bx_{2})\big)\rho_{0}(\bx_{1},\bx_{2})$.
The denominator of the left-hand side of \eqref{rho_inequality} is, by \eqref{eq:posttrade-wealth} and (i), greater than or equal to 
$(1-\eM)\big(w(\bx_{1})\wedge w(\bx_{2})\big)$, and the $l$-components are unchanged. Combining the three observations gives (iv).

(v) Note that the two post-trade states share the same $l$ and $J_{3}$, while by
\eqref{eq:master}
\[
|J_{1}(\bx,z)-J_{2}(\bx,z)|
=(\ue+\be)\,e^{l}|\theta(l,x,y,z)|
=(\ue+\be)\,\frac{1+\be\pi}{1+\be}\,w(\bx)\,|z-b(\bx)|
\le 2\,w(\bx)\,|z-b(\bx)|.
\]
Dividing by the denominator, which by \eqref{eq:posttrade-wealth} and (i) is greater than or equal to $(1-\eM)w(\bx)$, gives (v).

(vi) Write $b(\mathbf{x})=B\bigl(l,\pi(\mathbf{x})\bigr)$ with 
\begin{equation*}
  B(l,\pi):=\frac{\bigl(1+\be(e^{l})\bigr)\pi}{1+\be(e^{l})\pi},
  \qquad (l,\pi)\in\mathbb{R}\times[0,1].
\end{equation*}
Then $B$ takes values in $[0,1]$ and
\begin{equation}\label{eq:B-derivatives}
  \frac{\partial B}{\partial \pi} =\frac{1+\be}{(1+\be\pi)^{2}} \in (0,1+\eM] \subseteq (0,2], \qquad 
  \bigg| \frac{\partial B}{\partial l} \bigg| =\bigg| \frac{d}{dl}\be(e^{l}) \bigg|\,\frac{\pi\,(1-B)}{1+\be\pi} \le K_{\be},
\end{equation}
so $B$ is bounded by $1$ and Lipschitz on $\mathbb{R} \times [0,1]$ with constant $\max\{2,K_{\be}\}$. Applying (iv) of \lemref{lem:composition} to $A(l,\pi,r,z):=B(l,\pi)$ proves (vi).

(vii), (viii) By \eqref{eq:posttrade-wealth}, the post-trade to pre-trade wealth ratios appearing in $U_1,U_2$ are
$1-\be q$ and $1+\ue q$ respectively, where $q=q(l,\pi,z)$, and by (i) both are bounded away from $0$ by $(1+\eM)^{-1}\wedge(1-\eM)=1-\eM$.
Hence, setting
\begin{equation*}
  g_1:=\ln\bigl(1-\be q\bigr),
  \qquad
  g_2:=\ln\bigl(1+\ue q\bigr),
\end{equation*}
both take values in the compact interval
$I_{\eM}:=\bigl[\ln(1-\eM),\,\ln(1+\eM)\bigr]$,
and with $r=r(\bx)=\ln w(\bx)$ we may factor
\[
  U_i(\bx,z)=e^{F(r+g_i)-F(r)},\qquad i\in\{1,2\},
\]
where, as before, $F(r)=\ln \cU(e^{r})$. Since $\|F'\|_\infty\le p$, the map $(r,g)\mapsto e^{F(r+g)-F(r)}$ is bounded by $c_{\cU}$ on $\R\times I_{\eM}$ (cf.\ (iii)) and has partial derivatives
$\bigl(F'(r+g)-F'(r)\bigr)e^{F(r+g)-F(r)}$ and $F'(r+g)\,e^{F(r+g)-F(r)}$, both bounded by
$2p\,c_{\cU}$; hence it is Lipschitz on $\R\times I_{\eM}$ with constant $2p\,c_{\cU}$.
It remains to bound the partial derivatives of $(l,\pi,z)\mapsto g_i$. From
$q(l,\pi,z)=\frac{(1+\be\pi)(z-B(l,\pi))}{1+\be}$, $|z-B|\le 1$,
$1+\be\pi\le 2$, $1+\be\ge 1$, and the bounds \eqref{eq:B-derivatives}, we have 
\begin{align*}
  \bigg| \frac{\partial q}{\partial z} \bigg|
    &=\frac{1+\be\pi}{1+\be}\le 2, \qquad \qquad
  \bigg| \frac{\partial q}{\partial \pi} \bigg|
    \le \frac{1}{1+\be} \bigg( \be\,|z-B|+(1+\be\pi) \bigg| \frac{\partial B}{\partial \pi} \bigg| \bigg)
     \le 5,\\
  \bigg| \frac{\partial q}{\partial l} \bigg|
    & \le \frac{1}{1+\be} \bigg( \pi K_{\be} |z-B|+(1+\be\pi) \bigg| \frac{\partial B}{\partial l} \bigg| \bigg)
       +\frac{(1+\be\pi)|z-B|\,K_{\be}}{(1+\be)^{2}}
     \le K_{\be}+2K_{\be}+2K_{\be}=5K_{\be}.
\end{align*}
Since $|\be q|\vee|\ue q|\le \eM$ (as $|q|\le 1$ by \eqref{eq:master}) and, by (i),
the arguments $1-\be q$ and $1+\ue q$ of the logarithm are bounded away from $0$
by $1-\eM$, the chain rule yields
\[
  \bigg| \frac{\partial g_i}{\partial z} \bigg|
  +\bigg| \frac{\partial g_i}{\partial \pi} \bigg|
  +\bigg| \frac{\partial g_i}{\partial l} \bigg|
  \;\le\; c \big( \eM,K_{\ue},K_{\be} \big)<\infty,
  \qquad i\in\{1,2\},
\]
where in $\frac{\partial g_i}{\partial l}$ the explicit $l$-dependence of $\be(e^{l})$, $\ue(e^{l})$ contributes at most $(K_{\be}\vee K_{\ue})(1-\eM)^{-1}$. Consequently, $U_i(\bx,z)=A_i(l,\pi(\bx),r(\bx),z)$ for functions $A_1,A_2$ that are bounded by $c_{\cU}$ and Lipschitz on $\R\times[0,1]\times\R\times[0,1]$, with constant $c'=c'(p,\eM,K_{\ue},K_{\be})$. Items (vii) and (viii) now follow from (iv) and (v) of \lemref{lem:composition} applied to $A_{i}$ with $c_{A}=c_{\cU}\vee c'$. Enlarging $\hat{C}'$ so that it dominates all of the above constants completes the proof.
\end{proof}

\begin{lemma}[Coefficient estimates]    \label{lem:coeff}
    Recall $\tilde\mu,\gamma$ from \eqref{eq:mu_tilde}, \eqref{eq:gamma_define}, Assumption \ref{U_assumption}. There exists a constant $C_*>0$, depending only on the market coefficient bounds and Lipschitz constants, $\epsilon_M$, $p$, $\delta_U$, and the bounded H\"older norm of $r\mapsto U(e^r)/\mathcal{U}(e^r)$, such that for all $t',t''\in[0,T]$, $l,l',l'' \in \R$, and $\bx_{1}, \bx_{2} \in \R \times \D$ with $\bp_{i}=(x_{i},y_{i})$, 
    \begin{enumerate}
        \item[(i)] $\|\tilde\mu\|_{\infty}+\|\gamma\|_{\infty}\le C_{*}$;
        \item[(ii)] $|\tilde\mu(t',l',\bp_{1})-\tilde\mu(t'',l'',\bp_{1})|
                    +|\gamma(t',l',\bp_{1})-\gamma(t'',l'',\bp_{1})|
                    \le C_{*}(|t'-t''|+|l'-l''|)$;
        \item[(iii)] for $h\in\{\tilde\mu,\gamma\}$,
        \begin{equation}\label{eq:param}
            |h(t,l,\bp_{1})-h(t,l,\bp_{2})| \ \le\ C_{*}\,\min\!\big\{1,\ e^{2|l-l_{2}|}\rho_{0}(\bx_{1},\bx_{2})\big\}.
        \end{equation}
        \item[(iv)] $\|g\|_{\infty}\le C_{*}$ and 
        \begin{equation}\label{eq:terminal-ratio}
            |g(\bx_{1})-g(\bx_{2})|\ \le\ C_{*}\,[\rho_{0}(\bx_{1},\bx_{2})]^{\delta_{U}} \qquad \text{whenever } \rho_{0}(\bx_{1},\bx_{2})\le 1.
        \end{equation}
        In particular, $[g]_{\sn} < \infty$.
    \end{enumerate}
\end{lemma}

\begin{proof}
With $r:=\ln w$ and $\pi$, the identities $w\cU'(w)/\cU(w)=F'(r)$ and $w^{2}\cU''(w)/\cU(w)=F''(r)-F'(r)+F'(r)^{2}$ give
\[
    \tilde\mu=\mu-\tfrac{\sigma^{2}}{2}+\sigma^{2}\,\pi F'(r), \qquad \gamma=\tfrac{\sigma^{2}}{2}\big[\pi F'(r)+\pi^{2}\big(F''(r)-F'(r)+F'(r)^{2}\big)\big]+\big(\mu-\tfrac{\sigma^{2}}{2}\big)\pi F'(r),
\]
where $\mu,\sigma$ are evaluated at $(t,e^{l})$.
Thus each coefficient is a finite sum of terms of the form $a(t,l)H(\pi,r)$, where $a$ is bounded and Lipschitz continuous in $(t,l)$ and $H$ is smooth on $[0,1]\times\mathbb{R}$ with bounded first derivatives from \eqref{eq:Fbounds}:
\begin{equation}\label{eq:Hbounds}
    \|H\|_{\infty} + \bigg\| \frac{\partial H}{\partial \pi} \bigg\|_{\infty} + \bigg\| \frac{\partial H}{\partial r} \bigg\|_{\infty} \le c(p) < \infty.
\end{equation}

(i) It is straightforward due to \eqref{eq:Hbounds}.

(ii) Apply Lemma~\ref{lem:composition}(iii) to each function $H_j$ in the finite-sum representation. Combining the resulting bounds with the boundedness and Lipschitz continuity of the coefficients $a_j$, and summing over the finitely many terms, gives (ii).

(iii) Fix $(t,l)$ and write $h(t,l,\mathbf{p})=a_0(t,l)+\sum_{j=1}^{2}a_j(t,l)H_j(\pi,r)$, where $a_{0}, a_{1}, a_{2}$ are bounded and Lipschitz continuous in $(t,l)$ and $H_{1}, H_{2}$ are smooth on $[0,1]\times\mathbb{R}$ with bounded first derivatives. The term $a_0$ cancels when comparing $\mathbf{p}_1$ and $\mathbf{p}_2$. Applying \eqref{eq:comp-min} to each $H_j$ and summing the resulting estimates proves \eqref{eq:param}.

(iv) Since $x+[1-\ue(e^{l})]ye^{l} = w(\bx)\big(1-\ue(e^{l})\pi\big)$ with $1-\ue(e^{l})\pi\in[1-\eM,1]$, setting $\zeta(l,\pi):=\ln\big(1-\ue(e^{l})\pi\big)\in[\ln(1-\eM),0]$ we may factor
\[
    g(\bx)=\Phi\big(r+\zeta(l,\pi)\big)\,e^{\,F(r+\zeta(l,\pi))-F(r)}, \qquad \Phi(\varrho):=\frac{U(e^{\varrho})}{\cU(e^{\varrho})},
\]
with $r=\ln w(\bx)$ and $\pi=\pi(\bx)$ as before. We record the bounds on the three ingredients. First, Assumption \ref{U_assumption} precisely states that 
\[
    \| \Phi \|_{\infty} + [ \Phi ]_{\delta_{U}} < \infty, \qquad \text{where} \quad [ \Phi ]_{\delta_{U}} := \sup_{\varrho_{1} \neq \varrho_{2}} \frac{|\Phi(\varrho_{1}) - \Phi(\varrho_{2})|}{|\varrho_{1} - \varrho_{2}|^{\delta_{U}}}.
\]
Second, the map $(r,\zeta)\mapsto F(r+\zeta)-F(r)$ takes values in
$[-p|\zeta|, p|\zeta|]\subseteq[p\ln(1-\eM),\,-p\ln(1-\eM)]$ and has partial
derivatives $F'(r+\zeta)-F'(r)$ and $F'(r+\zeta)$, bounded by $2p$; hence
$(r,\zeta)\mapsto e^{F(r+\zeta)-F(r)}$ is Lipschitz and bounded by $(1-\eM)^{-p}$. Third, since $1-\ue\pi\ge1-\eM$,
\[
  \bigg| \frac{\partial \zeta}{\partial \pi} \bigg| = \frac{\ue}{1- \ue \pi}\le\frac{\eM}{1-\eM}, \qquad \bigg| \frac{\partial \zeta}{\partial l} \bigg| = \frac{\pi\,\big|\frac{d}{dl}\ue(e^{l})\big|}{1-\ue\pi} \le \frac{K_{\ue}}{1-\eM}.
\]
Consequently $g(\bx)=A(l,\pi,r)$ for a function $A$ that is bounded and $\delta_{U}$-H\"older continuous on $\R \times[0,1] \times \R$, with constant $c'$ depending only on $p, \, \eM, \, \| \Phi \|_{\infty} + [ \Phi ]_{\delta_{U}}, \, K_{\ue}$. Now (iv) of \lemref{lem:composition}, applied to $(l,\pi,r,z)\mapsto A(l,\pi,r)$ with $c_{A}=\|g\|_{\infty}\vee c'$, yields \eqref{eq:terminal-ratio}. Finally, $[g]_{\sn}<\infty$ follows since $g$ does not depend on $t$ and, $\rho\le 1$ implies $|g(\bx_{1})-g(\bx_{2})|\le C_{*}\rho_{0}^{\delta_{U}}\le C_{*}\rho^{\delta_{U}} \le C_{*}\rho^{\delta}$.
\end{proof}

\begin{lemma}[From the seminorm to moduli of continuity]\label{lem:seminorm}
Let $f : [0,T] \times \R \times \D \to \R$ satisfy
$\|f\|_{\infty}<\infty$ and $[f]_{\sn}<\infty$. Then:
\begin{enumerate}
\item[(i)] for all pairs $P_{1},P_{2} \in [0,T] \times \R \times \D$,
\begin{equation}\label{eq:patch}
|f(P_{1})-f(P_{2})|\ \le \big([f]_{\sn}+2\|f\|_{\infty}\big)\,
\min\{1,\rho(P_{1},P_{2})\}^{\delta};
\end{equation}
\item[(ii)] for any $(x,y) \in \D$, $f(\cdot,\cdot,x,y) \in C^{\frac{\delta}{2},\delta}([0,T] \times \R)$, with 
\[
\sup_{(x,y)\in \D}\,[f(\cdot,\cdot,x,y)]_{\frac\delta2,\delta}
\ \le\ [f]_{\sn}+2\|f\|_{\infty}.
\]
\end{enumerate}
\end{lemma}

\begin{proof}
(i) If $\rho\le1$ use the definition of $[\cdot]_{\sn}$, otherwise
$|f(P_{1})-f(P_{2})|\le2\|f\|_{\infty}=2\|f\|_{\infty}\min\{1,\rho\}^{\delta}$.

(ii) Fix $(x,y)$ and two points $(t_{1},l_{1}),(t_{2},l_{2})$.  Since the
weighted numerator in $\rho_{0}$ vanishes for equal $(x,y)$,
$\rho=\sqrt{|t_{1}-t_{2}|}+|l_{1}-l_{2}|$ exactly.
If $\rho\le1$, subadditivity $(a+b)^{\delta}\le a^{\delta}+b^{\delta}$ gives
$|\Delta f|\le[f]_{\sn}\,\rho^{\delta}
\le[f]_{\sn}(|t_{1}-t_{2}|^{\delta/2}+|l_{1}-l_{2}|^{\delta})$.
If $\rho>1$, then
$|\Delta f|\le2\|f\|_{\infty}\le2\|f\|_{\infty}
(|t_{1}-t_{2}|^{\delta/2}+|l_{1}-l_{2}|^{\delta})$.
\end{proof}

\begin{lemma}\label{lem:random}
Let $L^{(t,\bx)}$ be given by \eqref{L_process}.
\begin{enumerate}
\item[(i)] For any $k \ge 0$,
\[
\sup_{(t,\bx)\in[0,T]\times\mathbb{R}\times \D}
\mathbb{E}\bigg[\exp\Big(k\sup_{s\in[t,T]}\big|L^{(t,\bx)}_{s}-l\big|\Big)\bigg]
=:C_{k}<\infty,
\]
with $C_{k}$ depending only on $k,\|\tilde\mu\|_{\infty},\|\sigma\|_{\infty},T$.
\item[(ii)] Suppose $\rho((t_{1},\bx_{1}),(t_{2},\bx_{2})) \le 1$ with $t_{1} \le t_{2}$. Then for any $u \in [t_{2},T]$,
\[
\mathbb{E}\Big[\min\big\{1,\
\rho_{0}\big((L^{1}_{u},x_{1},y_{1}),(L^{2}_{u},x_{2},y_{2})\big)
\big\}^{\delta}\Big] \le C\,\rho((t_{1},\bx_{1}),(t_{2},\bx_{2}))^{\delta},
\]
where $C$ depends only on $\delta$, the constant of \eqref{L_expectation_Holder} and $C_{4}$ from (i). Consequently, for $f$ as in \lemref{lem:seminorm}, 
\begin{equation}\label{eq:randomf}
\mathbb{E}\big|f(u,L^{1}_{u},x_{1},y_{1})-f(u,L^{2}_{u},x_{2},y_{2})\big| \le C\big([f]_{\sn}+2\|f\|_{\infty}\big)\,
\rho((t_{1},\bx_{1}),(t_{2},\bx_{2}))^{\delta}.
\end{equation}
\end{enumerate}
\end{lemma}

\begin{proof}
(i) Write $L_{s}-l=\int_{t}^{s}\tilde\mu\,da+M_{s}$ with
$M_{s}:=\int_{t}^{s}\sigma(a,e^{L_{a}})\,dB_{a}$, so $\sup_{s \in [t,T]}|L_{s}-l|\le\|\tilde\mu\|_{\infty}T+\sup_{s \in [t,T]}|M_{s}|$. Since $\langle M\rangle_{T}\le\|\sigma\|_{\infty}^{2}T$ a.s., the time-change of exponential martingale inequality gives $\mathbb{P}(\sup_{s \in [t,T]}|M_{s}|\ge a) \le 2\exp(-a^{2}/(2\|\sigma\|_{\infty}^{2}T))$ for $a \ge 0$. Hence, by the tail-integration formula,
\begin{equation} \label{Tail_integration}
    \mathbb{E}\Big[\exp\Big(k\sup_{s \in [t,T]}|M_{s}|\Big)\Big] = 1 + k\int_{0}^{\infty}e^{ka}\,\mathbb{P}\Big(\sup_{s \in [t,T]}|M_{s}|\ge a\Big)\,da \le 1 + 2k\int_{0}^{\infty}e^{ka-a^{2}/(2\|\sigma\|_{\infty}^{2}T)}\,da<\infty,
\end{equation}
and combining with the drift bound gives $C_{k}\le e^{k\|\tilde\mu\|_{\infty}T}\big(1 + 2k\int_{0}^{\infty}e^{ka-a^{2}/(2\|\sigma\|_{\infty}^{2}T)}da\big)<\infty$; note that $\|\tilde\mu\|_{\infty}\le C_{*}$ by (i) of \lemref{lem:coeff}.

(ii) Set $\Xi_{u}:=|L^{1}_{u}-l_{1}|+|L^{2}_{u}-l_{2}|$ and denote by 
\[
\rho_{0}^{\mathrm{rand}}(u)
:=\rho_{0}\big((L^{1}_{u},x_{1},y_{1}),(L^{2}_{u},x_{2},y_{2})\big)
=\big|L^{1}_{u}-L^{2}_{u}\big|
+\frac{|x_{1}-x_{2}|+e^{L^{1}_{u}\vee L^{2}_{u}}|y_{1}-y_{2}|}
       {\big(x_{1}+y_{1}e^{L^{1}_{u}}\big)\wedge\big(x_{2}+y_{2}e^{L^{2}_{u}}\big)}.
\]
Since $e^{L^{1}_{u}\vee L^{2}_{u}}\le e^{l_{1}\vee l_{2}}e^{\Xi_{u}}$ and
$x_{i}+y_{i}e^{L^{i}_{u}}\ge e^{-\Xi_{u}}\,w(\bx_{i})$,
\[
\rho_{0}^{\mathrm{rand}}
\ \le\ |L^{1}_{u}-L^{2}_{u}|
+e^{2\Xi_{u}}\,\rho_{0}(\bx_{1},\bx_{2}).
\]
Using $\min\{1,r\}^{\delta}\le r^{\delta}$, Jensen's inequality
($\delta<1$), \eqref{L_expectation_Holder}, Cauchy-Schwarz, and (i) of this lemma (take $k=4$), we get 
\[
\mathbb{E}\big[\min\{1,\rho_{0}^{\mathrm{rand}}\}^{\delta}\big]
\le\big(\mathbb{E}\,\rho_{0}^{\mathrm{rand}}\big)^{\delta}
\le\Big(\mathbb{E}|L^{1}_{u}-L^{2}_{u}|
        +\rho_{0}\,\mathbb{E}\,e^{2\Xi_{u}}\Big)^{\delta}
\le\big(C\rho+C_{4}\rho_{0}\big)^{\delta}
\le C\rho^{\delta}.
\]
The display \eqref{eq:randomf} follows by combining with \eqref{eq:patch}
at the equal-time pair $(u,L^{1}_{u},x_{1},y_{1})$ and
$(u,L^{2}_{u},x_{2},y_{2})$.
\end{proof}

\begin{lemma}\label{lem:pathfun}
Assume $\rho\le1$ and $0 \le t_{1} \le t_{2} \le T$. Then there is $C_{\alpha}>0$ (depending on $\alpha$) with
\begin{equation}\label{eq:weights}
    \sup_{u\in[t_{2},T]} \mathbb{E}\big|\mathcal{E}_{1}(u)-\mathcal{E}_{2}(u)\big| \le C_{\alpha}\,\rho^{\delta}.
\end{equation}
In particular, $[w_{1}]_{\sn}<\infty$.
\end{lemma}

\begin{proof}
First we easily have $0<\mathcal{E}_{i}(u)\le e^{\|\gamma\|_{\infty}T}$ for $u\in[t_{i},T]$. The integrand $\mathcal{E}_{i}$ of exponent term is bounded by $\alpha+\Lambda+\|\gamma\|_{\infty}$ so 
\[
|\ln(\mathcal{E}_{1})(u)-\ln(\mathcal{E}_{2})(u)|\le
\underbrace{(\alpha+\Lambda+\|\gamma\|_{\infty})\,|t_{1}-t_{2}|}_{\text{time shift}}
+\int_{t_{2}}^{u}\big(|\Delta\lambda_{a}|+|\Delta\gamma_{a}|\big)da,
\]
and the elementary inequality $|e^{a}-e^{b}|\le e^{a \vee b}|a-b|$ together with $\ln(\mathcal{E}_{i})(u)\le\|\gamma\|_{\infty}T$ give
\[
    \big|\mathcal{E}_{1}(u)-\mathcal{E}_{2}(u)\big| \le e^{\|\gamma\|_{\infty}T}\Big( (\alpha+\Lambda+\|\gamma\|_{\infty})\,|t_{1}-t_{2}| + \int_{t_{2}}^{u}\big(|\Delta\lambda_{a}|+|\Delta\gamma_{a}|\big)da\Big),
\]
where
\[
    \Delta\gamma_{a} =\underbrace{\gamma(a,L^{1}_{a},p_{1})-\gamma(a,L^{2}_{a},p_{1})}_{=:\,\Delta\gamma_{a}^{(l)}} +\underbrace{\gamma(a,L^{2}_{a},p_{1})-\gamma(a,L^{2}_{a},p_{2})}_{=:\,\Delta\gamma_{a}^{(\mathrm{param})}}.
\]
The upper bound of time shift term is $C_{\alpha}\rho^{2}\le C_{\alpha}\rho^{\delta}$. By Assumption \ref{Assumption_1} and (ii) of \lemref{lem:coeff}, $|\Delta\lambda_{a}|+|\Delta\gamma_{a}|^{(l\text{-part})} \le C|L^{1}_{a}-L^{2}_{a}|$, while by (iii) of \lemref{lem:coeff}, $|\Delta\gamma_{a}|^{(\text{param})} \le C_{*}\min\{1,e^{2|L^{2}_{a}-l_{2}|}\rho_{0}\}$. Taking expectations, the first contribution is $\le C\rho$ by \eqref{L_expectation_Holder}; the second is, by (i) of \lemref{lem:random}, $\le C_{*}\rho_{0}\,\mathbb{E}\,e^{2|L^{2}_{a}-l_{2}|}\le C\rho$. Since $\rho\le\rho^{\delta}$ for $\rho\le1$, \eqref{eq:weights} follows.

Finally, we check $[w_{1}]_{\sn}<\infty$. The representation of $w_{1}$ gives 
\begin{align*}
    |w_{1}(t_{1},\bx_{1})-w_{1}(t_{2},\bx_{2})| & \le e^{\alpha T} \|g\|_{\infty}\,\mathbb{E}\big|\mathcal{E}_{1}(T)-\mathcal{E}_{2}(T)\big| +e^{\alpha T + \|\gamma\|_{\infty}T}\,\mathbb{E}\big|g(L^{1}_{T},x_{1},y_{1})-g(L^{2}_{T},x_{2},y_{2})\big| \\
    & \le e^{\alpha T}C_{*}C_{\alpha}\,\rho^{\delta} + e^{(\alpha+\|\gamma\|_{\infty})T}C\big([g]_{\sn}+2\|g\|_{\infty}\big)\rho^{\delta}
\end{align*}
where the last inequality is derived from \eqref{eq:weights}, (iv) of \lemref{lem:coeff}, and \eqref{eq:randomf} applied to the time-independent function $g$, whose seminorm is finite by \eqref{eq:terminal-ratio}. Hence $[w_{1}]_{\sn}<\infty$.
\end{proof}

\begin{lemma}\label{lem:patchedK}
    Let $f \in C_{bc}([0,T] \times \R \times \D)$ with $[f]_{\sn} < \infty$. Then $\|\sup_{z \in [0,1]}K(f)\|_{\infty}\le c_{\cU}\|f\|_{\infty}$ and for all $(t_{1},\bx_{1}),(t_{2},\bx_{2}) \in [0,T] \times \R \times \D$,
    \begin{equation}\label{eq:patchedK}
        \bigg|\sup_{z\in[0,1]}K(f)(t_{1},\bx_{1},z) -\sup_{z\in[0,1]}K(f)(t_{2},\bx_{2},z)\bigg| \le C\big([f]_{\sn}+2\|f\|_{\infty}\big) \min\{1,\rho\}^{\delta}
    \end{equation}
    for a constant $C > 0$ independent of $f$. In particular, the conclusion holds for $f = w_{n}$, $n \in \N$, and for $f = \hat{w}$.
\end{lemma}

\begin{proof}
The bound $\|\sup_{z \in [0,1]}K(f)\|_{\infty}\le c_{\cU}\|f\|_{\infty}$ is immediate from (iii) of \lemref{lem:post_trade}. For \eqref{eq:patchedK}, if $\rho := \rho((t_{1},\bx_{1}),(t_{2},\bx_{2})) > 1$, then the left-hand side is at most $2c_{\cU}\|f\|_{\infty} \le C ([f]_{\sn}+2\|f\|_{\infty}) \min\{1,\rho\}^{\delta}$ for $C \ge c_{\cU}$, so we may and do assume $\rho \le 1$, hence also $\rho_{0} := \rho_{0}(\bx_{1},\bx_{2}) \le 1$. Assume moreover, without loss of generality, $b(\mathbf{x}_{1}) \le b(\mathbf{x}_{2})$; the other case is treated symmetrically. Then
\begin{align}
  & \bigg| \sup_{z \in [0,1]} K(f)(t_{1}, \mathbf{x}_{1}, z) - \sup_{z \in [0,1]} K(f)(t_{2}, \mathbf{x}_{2}, z) \bigg| \nonumber \\
  & \le \sup_{z \in [0,1]} \bigg| f \big( t_{1},l_{1}, J_{1}(\mathbf{x}_{1},z), J_{3}(\mathbf{x}_{1},z) \big) U_{1}(\mathbf{x}_{1},z) \mathbf{1}_{ \{ z \ge b(\mathbf{x}_{1}) \} }  \nonumber \\
  & \qquad \quad + f \big( t_{1},l_{1}, J_{2}(\mathbf{x}_{1},z), J_{3}(\mathbf{x}_{1},z) \big) U_{2}(\mathbf{x}_{1},z) \mathbf{1}_{ \{ z < b(\mathbf{x}_{1}) \} } \nonumber \\
  & \qquad \quad - f \big( t_{2},l_{2}, J_{1}(\mathbf{x}_{2},z), J_{3}(\mathbf{x}_{2},z) \big) U_{1}(\mathbf{x}_{2},z) \mathbf{1}_{ \{ z \ge b(\mathbf{x}_{2}) \} }  \nonumber \\
  & \qquad \quad - f \big( t_{2},l_{2}, J_{2}(\mathbf{x}_{2},z), J_{3}(\mathbf{x}_{2},z) \big) U_{2}(\mathbf{x}_{2},z) \mathbf{1}_{ \{ z < b(\mathbf{x}_{2}) \} } \bigg| \nonumber \\
  & \le \sup_{z \in [0,1]} \bigg| f \big( t_{1},l_{1}, J_{1}(\mathbf{x}_{1},z), J_{3}(\mathbf{x}_{1},z) \big) [U_{1}(\mathbf{x}_{1},z) - U_{1}(\mathbf{x}_{2},z)] \nonumber \\
  & \quad + [f \big( t_{1},l_{1}, J_{1}(\mathbf{x}_{1},z), J_{3}(\mathbf{x}_{1},z) \big) - f \big( t_{2},l_{2}, J_{1}(\mathbf{x}_{2},z), J_{3}(\mathbf{x}_{2},z) \big)] U_{1}(\mathbf{x}_{2},z) \bigg| \mathbf{1}_{ \{ z \ge b(\mathbf{x}_{2}) \} } \nonumber \\
  & \quad + \sup_{z \in [0,1]} \bigg| f \big( t_{1},l_{1}, J_{2}(\mathbf{x}_{1},z), J_{3}(\mathbf{x}_{1},z) \big) [U_{2}(\mathbf{x}_{1},z) - U_{2}(\mathbf{x}_{2},z)] \nonumber \\
  & \quad + [f \big( t_{1},l_{1}, J_{2}(\mathbf{x}_{1},z), J_{3}(\mathbf{x}_{1},z) \big) - f \big( t_{2},l_{2}, J_{2}(\mathbf{x}_{2},z), J_{3}(\mathbf{x}_{2},z) \big)] U_{2}(\mathbf{x}_{2},z) \bigg| \mathbf{1}_{ \{ z < b(\mathbf{x}_{1}) \} } \nonumber \\
  & \quad + \sup_{z \in [0,1]} \bigg| \big[ f \big( t_{1},l_{1}, J_{1}(\mathbf{x}_{1},z), J_{3}(\mathbf{x}_{1},z) \big) - f \big( t_{2},l_{2}, J_{1}(\mathbf{x}_{2},z), J_{3}(\mathbf{x}_{2},z) \big) \big] U_{1}(\mathbf{x}_{1},z) \nonumber \\
  & \quad + f \big( t_{2},l_{2}, J_{1}(\mathbf{x}_{2},z), J_{3}(\mathbf{x}_{2},z) \big) [U_{1}(\mathbf{x}_{1},z) - U_{1}(\mathbf{x}_{1},b(\mathbf{x}_{1})) + U_{2}(\mathbf{x}_{2},b(\mathbf{x}_{2})) - U_{2}(\mathbf{x}_{2},z)] \nonumber \\
  & \quad + \big[ f \big( t_{2},l_{2}, J_{1}(\mathbf{x}_{2},z), J_{3}(\mathbf{x}_{2},z) \big) - f \big( t_{2},l_{2}, J_{2}(\mathbf{x}_{2},z), J_{3}(\mathbf{x}_{2},z) \big) \big] U_{2}(\mathbf{x}_{2},z) \bigg| \mathbf{1}_{ \{ b(\mathbf{x}_{1}) \le z < b(\mathbf{x}_{2}) \} }. \label{K_Holder}
\end{align}
Note that the second bracket of the third supremum in \eqref{K_Holder} rests on the identity $U_{1}(\bx,b(\bx))=U_{2}(\bx,b(\bx))=1$ for every $\bx$, which holds because $\theta(l,x,y,b(\bx))=0$.

We now estimate the three supremum in \eqref{K_Holder} separately, writing $c_{f}:=[f]_{\sn}+2\|f\|_{\infty}$ and $\hat{C}'$ for the constant of \lemref{lem:post_trade}. On the regions $\{z \ge b(\mathbf{x}_{2})\}$ and $\{z < b(\mathbf{x}_{1})\}$: by (vii) of \lemref{lem:post_trade}, $|U_{i}(\mathbf{x}_{1},z)-U_{i}(\mathbf{x}_{2},z)| \le \hat{C}'\rho_{0} \le \hat{C}'\rho^{\delta}$ (using $\rho_{0} \le \rho \le \rho^{\delta}$ for $\rho \le 1$); by (iv) of \lemref{lem:post_trade}, the post-trade states satisfy
\[
\rho\big((t_{1},l_{1},J_{i}(\mathbf{x}_{1},z),J_{3}(\mathbf{x}_{1},z)),\,(t_{2},l_{2},J_{i}(\mathbf{x}_{2},z),J_{3}(\mathbf{x}_{2},z))\big)
\le \sqrt{|t_{1}-t_{2}|}+\hat{C}'\rho_{0} \le (1 \vee \hat{C}')\rho,
\]
so \eqref{eq:patch} gives $|f(t_{1},l_{1},J_{i}(\mathbf{x}_{1},z),J_{3}(\mathbf{x}_{1},z))-f(t_{2},l_{2},J_{i}(\mathbf{x}_{2},z),J_{3}(\mathbf{x}_{2},z))| \le (1 \vee \hat{C}')^{\delta} c_{f}\rho^{\delta}$. Together with $\|f\|_{\infty} \le c_{f}$ and $\|U_{i}\|_{\infty} \le c_{\cU}$ from (iii) of \lemref{lem:post_trade}, each of the first two supremum in \eqref{K_Holder} is bounded by $Cc_{f}\rho^{\delta}$.

On the middle region $\{b(\mathbf{x}_{1}) \le z < b(\mathbf{x}_{2})\}$: here $|z-b(\mathbf{x}_{i})| \le b(\mathbf{x}_{2})-b(\mathbf{x}_{1}) \le \hat{C}'\rho_{0}$ for $i \in \{1,2\}$ by (vi) of \lemref{lem:post_trade}, so (viii) of \lemref{lem:post_trade} gives
\[
|U_{1}(\mathbf{x}_{1},z)-U_{1}(\mathbf{x}_{1},b(\mathbf{x}_{1}))|+|U_{2}(\mathbf{x}_{2},b(\mathbf{x}_{2}))-U_{2}(\mathbf{x}_{2},z)| \le 2(\hat{C}')^{2}\rho_{0} \le 2(\hat{C}')^{2}\rho^{\delta}.
\]
The first bracket of the third supremum in \eqref{K_Holder} is estimated exactly as above via (iv) of \lemref{lem:post_trade} and \eqref{eq:patch}. For the last bracket, (v) of \lemref{lem:post_trade} gives 
\[
  \rho_{0}\big((l_{2},J_{1}(\mathbf{x}_{2},z),J_{3}(\mathbf{x}_{2},z)),(l_{2},J_{2}(\mathbf{x}_{2},z),J_{3}(\mathbf{x}_{2},z))\big) \le \hat{C}'|z-b(\mathbf{x}_{2})| \le (\hat{C}')^{2}\rho_{0},
\]
so \eqref{eq:patch} (at equal times $t_{2}$) yields 
\[
|f(t_{2},l_{2},J_{1}(\mathbf{x}_{2},z),J_{3}(\mathbf{x}_{2},z))-f(t_{2},l_{2},J_{2}(\mathbf{x}_{2},z),J_{3}(\mathbf{x}_{2},z))| \le (1 \vee \hat{C}')^{2\delta} c_{f}\rho^{\delta}.
\]
Combining the three regions and enlarging $C$ proves \eqref{eq:patchedK}.
\end{proof}

\begin{lemma}[Uniform square integrability] \label{lem:uniform-integrability}
    For every $(t,l,x,y)\in[0,T]\times\mathbb{R}\times\mathbb{D}$, 
    \begin{equation} \label{Square_integrable}
        \sup_{N\in\mathcal{N}(t,l,x,y)} \mathbb{E}_{t,l,x,y} \bigg[ \sup_{u \in [t,T]} \cU^{2} \big( W_{u}^{N} \big) \bigg] < \infty
    \end{equation}
    where $W_u^{N}$ denotes the total wealth under the admissible strategy $N$. In particular, the estimate \eqref{Square_integrable} holds for the original initial state.
\end{lemma}

\begin{proof}
The case $t=T$ is immediate, hence let $t \in [0, T)$, $s \in [t, T]$, and admissible strategy $N$ be given. For the sake of simplicity, we denote $W, W^{\mathrm{liq}}$ by the total wealth and liquidated total wealth equipped with admissible strategy $N$ respectively in this proof. First one can observe that 
\begin{equation} \label{W_property}
    (1 - \epsilon_{M}) W_{s} \le W_{s}^{\mathrm{liq}} \le W_{s} \quad \text{ and } \quad (1 - \epsilon_{M}) W_{s-} \le W_{s} \le W_{s-}
\end{equation}
and the dynamics \eqref{SDE_W0} and \eqref{eq:W1dyn} imply 
\begin{equation} \label{SDE_W}
  d W_{s} = \pi_{s-} W_{s-} [\mu(s,e^{L_{s}}) d s + \sigma(s,e^{L_{s}}) d B_{s}] - [\overline{\epsilon}(e^{L_{s}}) e^{L_{s}} N_{s}^{+} + \underline{\epsilon}(e^{L_{s}}) e^{L_{s}} N_{s}^{-}] d P_{s}
\end{equation}
where $\pi_{s} := \frac{W_{s}^{(1)}}{W_{s}} \in [0, 1]$. Apply It\^o's lemma to \eqref{SDE_W} to get 
\begin{equation} \label{W_expand}
  \begin{split}
  W_{s} & = W_{t} \exp \bigg( \int_{t}^{s} \Big[ \pi_{u-} \mu(u,e^{L_{u}}) - \frac{1}{2} \pi_{u-}^{2} \sigma^{2}(u,e^{L_{u}}) \Big] du + \int_{t}^{s} \pi_{u-} \sigma(u,e^{L_{u}}) d B_{u} \\
  & \qquad \qquad \qquad + \int_{t}^{s} \ln \Big( 1 - \frac{\overline{\epsilon}(e^{L_{u}}) e^{L_{u}} N_{u}^{+}}{W_{u-}} - \frac{\underline{\epsilon}(e^{L_{u}}) e^{L_{u}} N_{u}^{-}}{W_{u-}} \Big) dP_{u} \bigg).
  \end{split}
\end{equation}
Let $M_{s, t} := \int_{t}^{s} \pi_{u-} \sigma(u,e^{L_{u}}) d B_{u}$. Since the drift integrand of \eqref{W_expand} is bounded by $\| \mu \|_{\infty} + \| \sigma \|_{\infty}^{2}$, since by \eqref{W_property} the jump exponent in \eqref{W_expand} lies in $[(P_{T} - P_{t}) \ln (1 - \eM),\, 0]$, and since $e^{2 q M_{s, t}} \le e^{2 |q| |M_{s, t}|}$, we obtain from \eqref{W_expand}, for every $q \in \R$, 
\[
    \sup_{u \in [t,T]} W_{u}^{2 q} \le W_{t}^{2 q}\, e^{2 |q| (\| \mu \|_{\infty} + \| \sigma \|_{\infty}^{2})(T-t)}\, (1 - \epsilon_M)^{- 2 |q| (P_T - P_t)} \exp \Big( 2 |q| \sup_{u \in [t,T]} |M_{u, t}| \Big).
\]
Hence, by the Cauchy-Schwarz inequality,
\begin{equation}  \label{sup_inequality}
  \begin{split}
    \mathbb{E}_{t,l,x,y} \Big[ \sup_{u \in [t,T]} W_{u}^{2 q} \Big] & \le W_{t}^{2 q} e^{2 |q| (\| \mu \|_{\infty} + \| \sigma \|_{\infty}^{2}) (T-t)} \mathbb{E}_{t,l,x,y} \big[ (1 - \eM)^{- 4 |q| (P_{T} - P_{t})} \big]^{\frac{1}{2}} \\
    & \qquad \cdot \mathbb{E}_{t,l,x,y} \bigg[ \exp \Big( 4 |q| \sup_{u \in [t,T]} |M_{u, t}| \Big) \bigg]^{\frac{1}{2}}.
    \end{split}
\end{equation}
Let us first treat the expectation of \eqref{sup_inequality}. Conditional on the future intensity path $\{\lambda(u,e^{L_u}):u\in[t,T]\}$, the random variable $P_T-P_t$ is Poisson with parameter
\[
    \int_t^T\lambda(u,e^{L_u})\,du.
\]
Therefore, by the tower property,
\begin{align}
    \mathbb{E}_{t,l,x,y} \Big[ (1-\epsilon_M)^{-4|q|(P_T-P_t)} \Big] &= \mathbb{E}_{t,l,x,y} \bigg[ \exp\Big( \big[(1-\epsilon_M)^{-4|q|}-1\big] \int_t^T\lambda(u,e^{L_u})\,du \Big) \bigg]    \nonumber
    \\
    &\le \exp\Big( \Lambda(T-t) \big[(1-\epsilon_M)^{-4|q|}-1\big] \Big).   \label{Poisson_expectation}
\end{align}
For the second expectation of \eqref{sup_inequality}, since $\langle M_{\cdot,t}\rangle_T \le \|\sigma\|_\infty^2(T-t)$, the exponential maximal inequality for continuous martingales gives
\[
    \mathbb{P}_{t,l,x,y} \Big( \sup_{u\in[t,T]}|M_{u,t}|\ge a \Big) \le 2\exp\bigg( -\frac{a^2}{2\|\sigma\|_\infty^2(T-t)} \bigg), \qquad a\ge0,
\]
under the restarted law corresponding to $\mathbb{E}_{t,l,x,y}$, and as in \eqref{Tail_integration}
\begin{equation} \label{Stochastic_expectation}
    \mathbb{E}_{t,l,x,y} \bigg[ \exp \Big( 4 |q| \sup_{u \in [t,T]} |M_{u, t}| \Big) \bigg] \le 1 + 8 |q| \int_{0}^{\infty} \exp \bigg( 4 |q| a- \frac{a^{2}}{2\|\sigma\|_{\infty}^{2}(T-t)} \bigg) \,da.
\end{equation}

Now applying \eqref{Poisson_expectation} and \eqref{Stochastic_expectation} to \eqref{sup_inequality}, we have for any $q \in \R$, 
\begin{align*}
    \mathbb{E}_{t,l,x,y} \Big[ \sup_{u \in [t,T]} W_{u}^{2q} \Big] & \le W_{t}^{2 q} e^{2 |q| (\| \mu \|_{\infty} + \| \sigma \|_{\infty}^{2}) (T-t)} \exp \bigg( \frac{1}{2} \Lambda (T-t)[(1 - \eM)^{- 4 |q|} - 1] \bigg) \\
    & \quad \cdot \bigg( 1 + 8 |q| \int_{0}^{\infty} \exp \bigg( 4 |q| a- \frac{a^{2}}{2\|\sigma\|_{\infty}^{2}(T-t)} \bigg)\,da \bigg)^{\frac{1}{2}}.
\end{align*}
Since the Cauchy-Schwarz inequality implies $\cU^{2}(w) = (1 + w^{p} + w^{-p})^{2} \le 3 (1 + w^{2p} + w^{-2p})$, taking $q = p$ and $q = -p$ deduces the desired result.
\end{proof}

\bigskip

\section*{Funding}

Tae Ung Gang is supported by the National Research Foundation of Korea under Grant RS-2026-25468875, RS-2026-25487734, and RS-2025-00513609, funded by the Korean government (MSIT). Donghan Kim is supported by the National Research Foundation of Korea under Grant RS-2025-00513609, and by an Ewon Assistant Professorship Award at KAIST. 

\medskip

\section*{Conflict of Interest}
The authors declare no conflicts of interest.

\bigskip

\bibliographystyle{apacite}
\bibliography{references}

\end{document}